\documentclass[11pt,fleqn]{article}
\usepackage{amssymb}
\usepackage{amsmath}
\usepackage{amsfonts}
\usepackage{latexsym}
\usepackage{enumerate}
\usepackage{amsthm}
\usepackage{graphicx}
\usepackage{caption}
\usepackage{subcaption}
\usepackage{epsfig,pspicture}

\newtheorem{Theorem}{Theorem}

\newtheorem{Corollary}{Corollary}
\newtheorem{Proposition}{Proposition}
\theoremstyle{definition}
\newtheorem{defn}{Definition}
\newtheorem{rem}{Remark}
\newtheorem{prop}{Property}
\newcommand\X[1]{$\mathcal{#1}^{\bullet}$}

\newcommand{\hh}[1]{$ #1= (#1^{\bullet},#1^{\times} $)}
\newcommand{\vv}[1]{${#1}^{\bullet}$}
\newcommand\ee[1]{${#1}^{\times}$}

\title{Morphological filtering  on hypergraphs}
\author{
\small {Bino~Sebastian~V\thanks{On deputation from Mar Athanasius College, Kothamangalam, and is supported by the University Grants Commission(UGC), Govt. of India under the FDP scheme.}}\\
        \small {Research~Scholar}\\
        \small{Department of Computer~Applications}\\
        \small{Cochin University of Science and Technology}\\
        \small{binosebastianv@gmail.com} \and
        \small{A Unnikrishnan}\\
         \small{Dean, R~\&~D}\\
           \small{Rajagiri School of Engineering and Technology}\\
           \small{unnikrishnan\_a@live.com}\and
        \small{Kannan Balakrishnan}\\
         \small{Associate~Professor}\\
           \small{Department of Computer~Applications}\\
           \small{Cochin University of Science and Technology}\\
           \small{mullayilkannan@gmail.com}\and
        \small{Ramkumar P. B.}\\
         \small{Assistant~Professor}\\
           \small{Department of Mathematics}\\
           \small{Rajagiri School of Engineering and Technology}\\
           \small{ramkumar\_pb@rajagiritech.ac.in}
           }
           \date{ }

\begin{document}
\maketitle

\begin{abstract}
The focus of this article is to develop computationally efficient mathematical morphology operators on hypergraphs. To this aim we consider lattice structures on hypergraphs on which we build morphological operators. We develop a pair of  dual adjunctions between the vertex set and the hyper edge set of a hypergraph $H$, by defining a vertex-hyperedge correspondence. This allows us to recover the classical notion of a dilation/erosion of a subset of vertices and to extend it to subhypergraphs of $H$. Afterward, we propose several new openings, closings, granulometries and alternate sequential filters acting (i) on the subsets of the vertex and hyperedge set of $H$ and (ii) on the subhypergraphs of a hypergraph.

%The focus of this article is to develop computationally efficient mathematical morphology operators on hypergraphs. To this aim we consider lattice structures on hypergraphs on which we build morphological operators. We develop a pair of  dual adjunctions between the vertex set and the hyper edge set of a hypergraph $H$, by defining a vertex-hyperedge correspondence. This allows us to recover the classical notion of a dilation/erosion of a subset of vertices and to extend it to subhypergraphs of $H$. This paper also study the concept of morphological adjunction on hypergraphs for which both the input and the output are hypergraphs.
\end{abstract}


\section{Introduction}
Mathematical morphology, appeared in 1960s is a theory of nonlinear information processing \cite{gonzalez2002woods}, \cite{serra1988image}, \cite{serra1982image}, \cite{shih2010image}. It is a branch of image analysis based on algebraic, set-theoretic and geometric principles \cite{heijmans1990algebraic}, \cite{ronse1990mathematical}. Originally, it is devoleped for binary images by Matheron and Serra. They are the first to observe that a general theory of mathematical morphology is based on the assumption that the underlying image space is a complete lattice. Most of the morphological theory at this abstract level was developed and presented without making references to the properties of the underlying space. Considering digital objects carrying structural information, mathematical morphology has been developed on graphs \cite{cousty2013morphological}, \cite{cousty2009some}, \cite{heijmans1993graph}, \cite{vincent1989graphs} and simplicial complexes \cite{dias2011some}, but little work has been done on hypergraphs \cite{bloch2011mathematical}, \cite{bloch2013mathematical}, \cite{bloch2013similarity}, \cite{stell2012relations}.
 
When dealing with a hypergraph $H$, we need to consider the hypergraph induced by the subset $X^\bullet$ of vertices of $H$ (see figure 1(a) and (b), where the blue vertices and edges in (b) represents $X$). We associate with $X^\bullet$ the largest subset of hyperedges of $H$ such that the obtained pair is a hypergraph. We denote it by $H(X^\bullet)$ (see section 2.1 and figure 1(b)). We also consider a hypergraph induced by a subset $X^\times$ of the edges of $H$, namely $H(X^\times)$.

Here we propose a systematic study of the basic operators that are used to derive a set of hyper edges from a set of vertices and a set of vertices from a set of hyperedges. These operators are the hypergraph extension to the operators defined by Cousty \cite{cousty2013morphological}, \cite{cousty2009some} for graphs. Since a hypergraph becomes a graph when $|v(e)|=2$  for every hyperedge $e$, all the properties of these operators are satisfied for graphs also. We emphasis that the input and output of these operators are both hypergraphs. The blue subhypergraph in figure 1(c) is the result of the dilation $[\triangle, \delta](X)$ of the blue subhypergraph $X$ in figure 1(b) proposed in this paper. Here the resultant subhypergrah in figure 1(c) is not induced by its vertex set.

\begin{center}
\begin{figure}[htp!]
\begin{minipage}[b]{0.3\linewidth}
\centering
\includegraphics[scale=1]{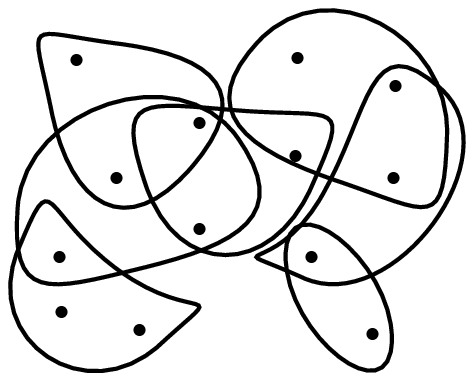}
\caption*{(a) $H$}
\label{fig:figure1}
\end{minipage}
\hspace{0cm}
\begin{minipage}[b]{0.3\linewidth}
\centering
\includegraphics[scale=1]{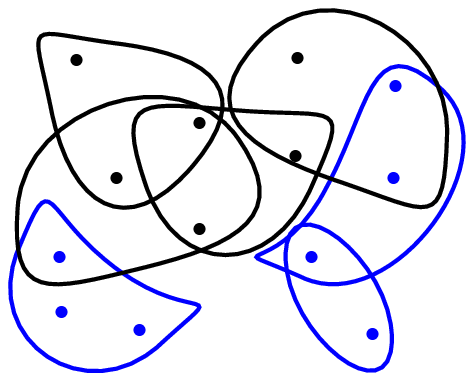}
\caption*{(b) $X$}
\label{fig:figure2}
\end{minipage}
\hspace{0cm}
\begin{minipage}[b]{0.3\linewidth}
\centering
\includegraphics[scale=1]{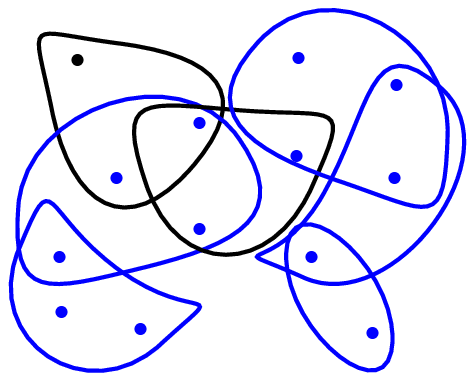}
\caption*{(c) Dilation $[\triangle, \delta](X)$}
%\label{fig:figure2}
\end{minipage}
\hspace{0cm}
%\label{fig1}
\caption{Illustration of hypergraph dilation}\label{fig:fig1}
\end{figure}
\end{center}

This paper is organized as follows. In Section 2 we recall some related works on graphs and hypergraphs. In Section 3, we describe some preliminary definitions and results on mathematical morphology and hypergraphs. In section 4, we define the vertex-hyperedge correspondence along with various dilations, erosions and adjunctions on hypergraphs. The properties of these morphological operators are studied in this section. In section 5, we propose several new openings, closings, granulometries and alternate sequential filters acting (i) on the subsets of vertices and hyperedges and (ii) on the subhypergraphs. Section 6 concludes the paper. 

\section{Related works}
Graph theoretic methods have found increasing applications in image analysis. Morphological operators are well studied on graphs. Vincent \cite{vincent1989graphs} defined morphological operators on a graph $G=(V,E)$, where $V$ represents a set of weighted vertices and $E$, a set of edges between vertices. The dilation (resp. erosion) replace the value of each vertex with the maximum (resp. minimum) value of its neighbors. Cousty et. al. \cite{cousty2013morphological}, \cite{cousty2009some} considered a graph as a pair $G=(G^\bullet, G^\times)$, where $G^\bullet$ is the set of vertices and $G^\times$ is the edge set of the graph $G$. They define morphological operators on various lattices formed by the graph $G$ by defining an edge-vertex correspondence. This powerful tool allows them to recover the clasical notion of a dilation/erosion of a subset of vertices of $G$. This lead them to propose several new openings, closings, granulometries and alternate sequential filters acting on the subsets of the edge sets, subsets of vertex sets and the lattice of subgraphs of $G$. These operators are further extended to functions that weight the vertices and edges of $G$ \cite{najman2012short} and are found to be useful in image filtering. In this work we aim to develop morphological operators on hypergraphs by defining a vertex-hyperedge correspondence. 

The theory of hypergraphs originated as a natural generalisation of graphs in 1960s. In a hypergraph, edges can connect any number of vertices and are called hyperedges. Considering the topological and geometrical aspects of an image, Bretto \cite{bretto1997combinatorics} has proposed a hypergraph model to represent an image. The theory of hypergraphs became an active area of research in image analysis \cite{bloch2013mathematical}, \cite{dharmarajan2010hypergraph}. The study of mathematical morphology operators on hypergraphs started recently, and little work being reported in this regard. Properties of morphological operators on hypergraphs are studied in \cite{stell2012relations}, in which subhypergraphs are considered as relations on hypergraphs. Recently, Bloch and Bretto \cite{bloch2011mathematical} introduced mathematical morphology on hypergraphs by forming various lattices on hypergraphs. Similarity and pseudo-metrics based on mathematical morphology are defined and illustrated in \cite{bloch2013similarity}. Based on these morphological operators, similarity measures are used for classification of data represented as hypergraphs  \cite{bloch2013mathematical}.

\section{Preliminaries}

\subsection{Hypergraphs}
We define a $hypergraph$ \cite{berge1989hypergraphs}, \cite{bloch2011mathematical} as a pair \hh{H} where \vv{H} is a set of points called $vertices$ and \ee{H} is composed of a family of subsets of \vv{H} called $hyperedges$. We denote \ee{H} by \ee{H}=$(e_{i})_{i \in I}$ where $I$ is a finite set of indices. The set of vertices forming the hyperedge $e$ is denoted by $v(e)$. A vertex $x$ in \vv{H} is called an $isolated~vertex$ of $H$ if $x \notin \underset{i \in I}\cup v(e_{i})$. The $empty~hypergraph$ is the hypergraph $H_{\phi}$ such that \vv{H}=$\phi$ and \ee{H}=$\phi$. The $partial~hypergraph$ $H'$ of $H$ generated by $J \subseteq I$ is the hypergraph \hh{H'} where \vv{H'} = \vv{H} and \ee{H'} = $(e_{j})_{j \in J}$.  A hypergraph $X=(X^\bullet,X^\times)$ is called a $subhypergraph$ of $H$, denoted by $X \subseteq H$, if $X^\bullet \subseteq H^\bullet$ and $X^\times \subseteq H^\times$.

Let \vv{X} $\subseteq$ \vv{H} and \ee{X} $\subseteq$ \ee{H} where \ee{X} $=(e_{j}), j \in J$ such that $J \subseteq I$. We denote by $\overline{X^\bullet}$ (resp. $\overline{X^\times}$) by the complementary set of \vv{X} (resp. \ee{X}). Let $H(X^\bullet)$ and $H(X^\times)$ respectively denote the hypergraphs $(X^\bullet, \lbrace e_{i}, i \in I |v(e_{i}) \subseteq X^\bullet \rbrace)$ and $(\underset{j \in J} \cup v(e_{j}), (e_{j})_{j \in J})$.

While dealing with a hypergraph $H$, we consider the subhypergraph induced by a subset \vv{X} of vertices of $H$ namely $H(X^\bullet)$, and the subhypergraph induced by a subset \ee{X} of hyperedges namely $H(X^\times)$. $H(X^\bullet)$ is the largest subhypergraph of $H$ with \vv{X} as vertex set and $H(X^\times)$ is the smallest subhypergraph of $H$ with \ee{X} as its hyperedge set.

\subsection{Mathematical Morphology}
Now let us briefly recall some algebraic tools that are fundamental in mathematical morphology \cite{cousty2013morphological}, \cite{heijmans1997composing}, \cite{heijmans1990algebraic}, \cite{ronse1990mathematical}. Given two lattices  $\mathcal{L}_{1}$ and  $\mathcal{L}_{2}$, any operator $\delta:\mathcal{L}_{1} \to \mathcal{L}_{2}$ that distributes over the supremum and preserves the least element is called a $dilation$ ($i.e.$ $\forall \varepsilon \subseteq \mathcal{L}_{1}, \delta(\vee_{1} \varepsilon) = \vee_{2} \lbrace \delta(X) | X \in \varepsilon \rbrace) $. Similarly an operator that distributes over the infimum and preserves the greatest element is called an $erosion$. 

Two operators $\epsilon:\mathcal{L}_{1} \to \mathcal{L}_{2}$ and $\delta:\mathcal{L}_{2} \to \mathcal{L}_{1}$ form an $adjunction$ $(\epsilon, \delta)$, if for any $X \in \mathcal{L}_{1} $ and any $Y \in \mathcal{L}_{2}$, we have $ \delta(X) \leq_{1} Y \Leftrightarrow X \leq_{2} \epsilon(Y)$, where $\leq_{1}$ and $\leq_{2}$ denote the order relations in $\mathcal{L}_{1}$ and $\mathcal{L}_{2}$ respectively \cite{heijmans1997composing}. Given two operators $\epsilon$ and  $\delta$, if the pair $(\epsilon, \delta)$ is an adjunction, then $\epsilon$ is an erosion and $\delta$ is a dilation. If $\leq_{1}$, $\leq_{2}$ and $\leq_{3}$ are three lattices and if $\delta: \mathcal{L}_{1} \to \mathcal{L}_{2}$, $\delta^{'}: \mathcal{L}_{2} \to \mathcal{L}_{3}$, $\epsilon: \mathcal{L}_{2} \to \mathcal{L}_{1}$ and $\epsilon^{'}: \mathcal{L}_{2} \to \mathcal{L}_{2}$ are four operators such that $(\epsilon, \delta)$ and $(\epsilon^{'}, \delta^{'})$ are adjunctions, then the pair $(\epsilon \circ \epsilon^{'}, \delta \circ \delta^{'})$ is also an adjunction.

Given two complemented lattices, $\mathcal{L}_{1}$ and $\mathcal{L}_{2}$, two operators $\alpha$ and $\beta$ are $dual$ with respect to the complement of each other, if for each $ X \in \mathcal{L}_{1}$, we have $\beta(X)=\overline{\alpha(\overline{X})}$. If $\alpha$ and $\beta$ are dual of each other, then $\beta$ is an erosion whenever $\alpha$ is a dilation.

\section{Hypergraph morphology: dilations, erosions and adjunctions}
In a hypergraph $H$, we can consider sets of points as well as sets of hyperedges. Therefore it is convenient to consider operators that go from one kind of sets to the other one. In this section we define such operators and study their morphological properties. Based on these operators, we propose several dilations, erosions and adjunctions on various lattices formed by $H$.

Hereafter the workspace (see \cite{cousty2013morphological} and \cite{cousty2009some} for a similar structure defined for graphs)  is a hypergraph \hh{H} and we consider the sets \vv{\mathcal{H}}, \ee{\mathcal{H}} and $\mathcal{H}$ of respectively all subsets of \vv{H}, all subsets of \ee{H} and all subhypergraphs of $\mathcal{H}$. 

The set $\mathcal{H}$ of all subhypergraphs of a hypergraph $H$ form a complete lattice \cite{stell2012relations}.  $\mathcal{H}$ is not a Boolean algebra as the complement of a subhypergraph of $H$ need not be a subhypergraph of $H$. But \vv{\mathcal{H}} and \ee{\mathcal{H}} are Boolean algebras. We define morphological operators on these lattices. For, we establish a correspondence between the vertex set and the hyperedge set of $H$. Composing these mappings produces morphological operators on the lattices \vv{\mathcal{H}}, \ee{\mathcal{H}} and $\mathcal{H}$.

\begin{defn}
(\textbf{Vertex-Hyperedge Correspondence})
We define the operators \vv{\delta}, \vv{\epsilon} from \ee{\mathcal{H}} into \vv{\mathcal{H}} and the operators \ee{\delta}, \ee{\epsilon} from \vv{\mathcal{H}} into \ee{\mathcal{H}} as follows.
\begin{center}
  \begin{tabular}{ | p{5cm} |  p{5cm} | p{5cm} | }
    \hline
     & \ee{\mathcal{H}} $\rightarrow$ \vv{\mathcal{H}} &  \vv{\mathcal{H}} $\rightarrow$ \ee{\mathcal{H}} \\
    \hline
    Provide the object with a hypergraph structure & \ee{X} $\rightarrow$ \vv{\delta}(\ee{X}) such that (\vv{\delta}(\ee{X}), \ee{X}) = $H(X^\times)$ & $X^\bullet \rightarrow \epsilon^\times(X^\bullet)$ such that $(X^\bullet, \epsilon^\times(X^\bullet)=H(X^\bullet)$ \\ \hline
    Provide its complement with a hypergraph structure & \ee{X} $\rightarrow$ \vv{\epsilon}(\ee{X}) such that ($\overline{\epsilon^\bullet({X^\times})}, \overline{X^\times}) = H(\overline{X^\times})$ & $X^\bullet \rightarrow \delta^\times(X^\bullet)$ such that $(\overline{X^\bullet}, \overline{\delta^\times(X^\bullet)}=H(\overline{X^\bullet})$ \\ 
    \hline    
  \end{tabular}
\end{center}
\end{defn}
These operators are illustrated in figures 2(a)-(f). The choice of $H$ is in such a way that every hyperedge of $H$ is incident with exactly four vertices, and the choice of $X$ is made to present a representative sample of the different possible configurations on subhypergraphs.
\begin{center}
\begin{figure}[htp!]
\begin{minipage}[b]{0.15\linewidth}
\centering
\includegraphics[scale=.5]{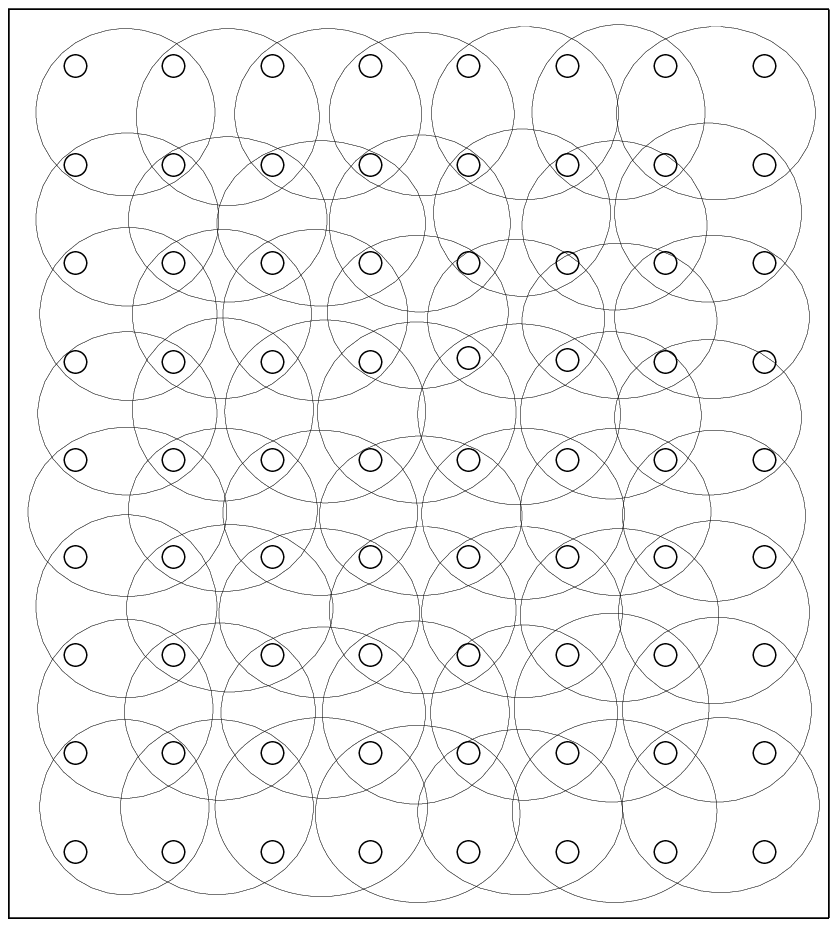}
\caption*{(a) $H$}
%\label{fig:figure1}
\end{minipage}
\hspace{2cm}
\begin{minipage}[b]{0.15\linewidth}
\centering
\includegraphics[scale=.5]{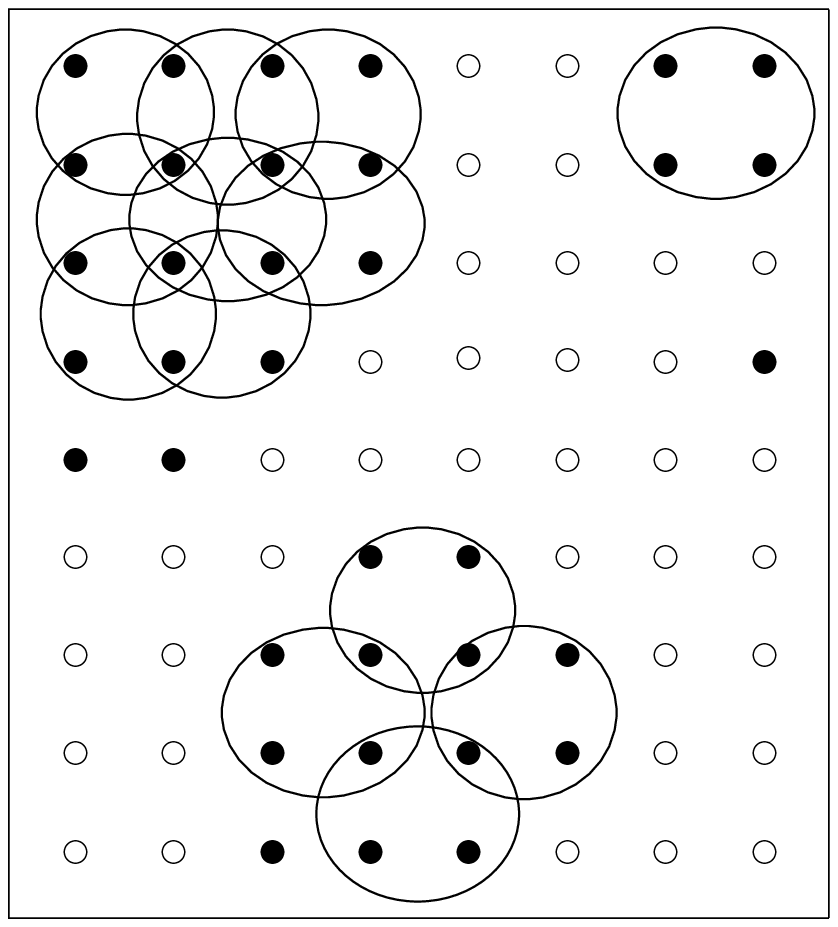}
\caption*{(b) $X$}
%\label{fig:figure2}
\end{minipage}
\hspace{2cm}
\begin{minipage}[b]{0.15\linewidth}
\centering
\includegraphics[scale=.5]{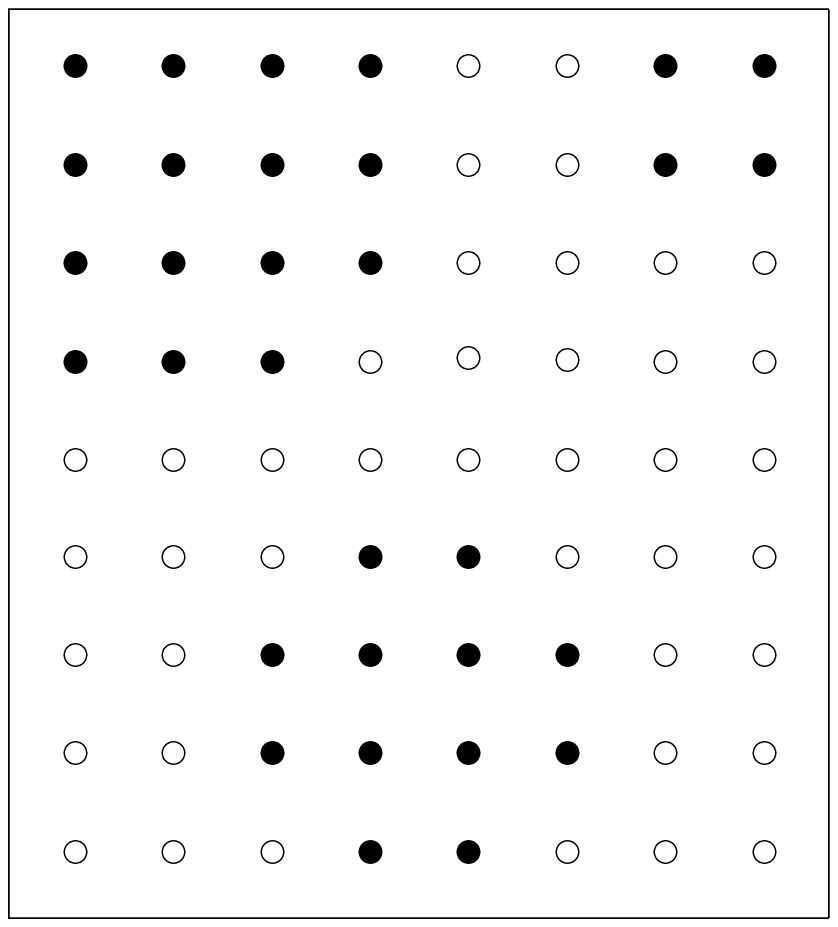}
\caption*{(c) $\delta^{\bullet}(X^{\times}) $}
%\label{fig:figure2}
\end{minipage}
\hspace{2cm}
\begin{minipage}[b]{0.15\linewidth}
\centering
\includegraphics[scale=.5]{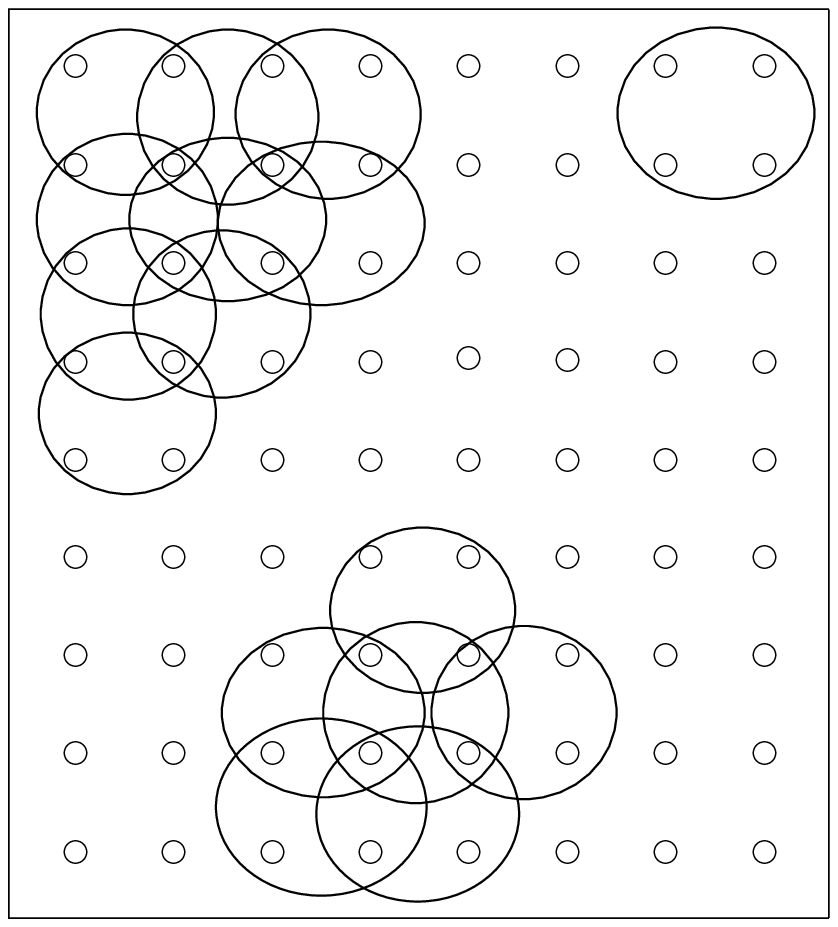}
\caption*{(d) $\epsilon^{\times}(X^{\bullet})$}
%\label{fig:figure2}
\end{minipage}
\hspace{2cm}
\begin{minipage}[b]{0.15\linewidth}
\centering
\includegraphics[scale=.5]{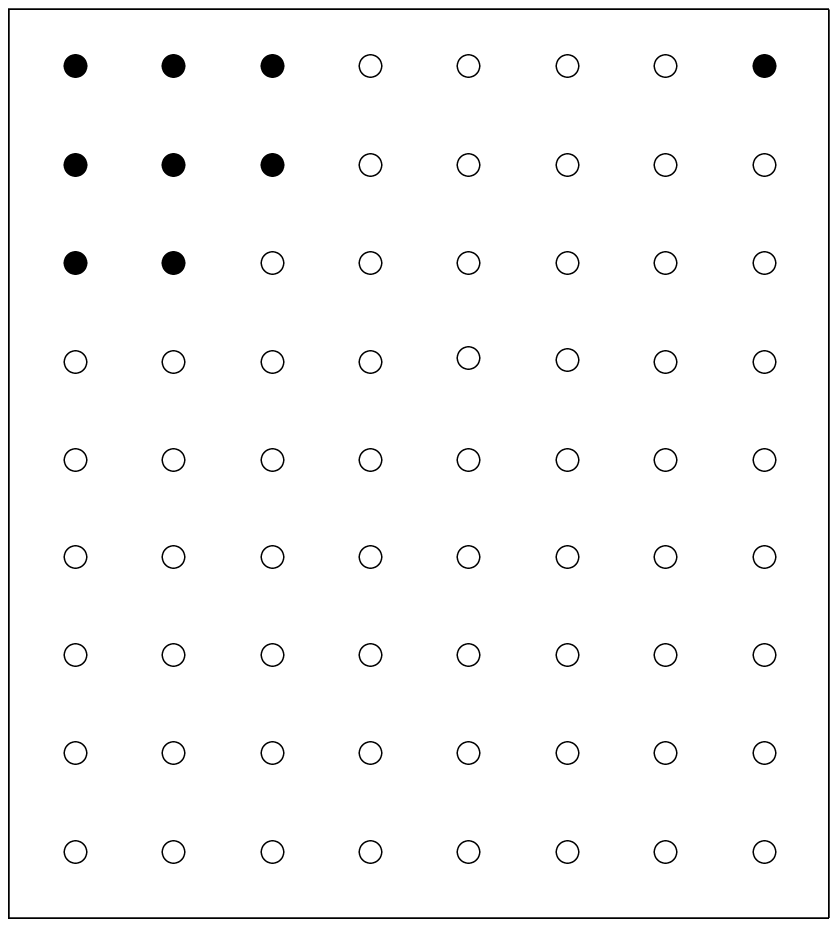}
\caption*{(e) $\epsilon^{\bullet}(X^{\times})$}
%\label{fig:figure2}
\end{minipage}
\hspace{2cm}
\begin{minipage}[b]{0.15\linewidth}
\centering
\includegraphics[scale=.5]{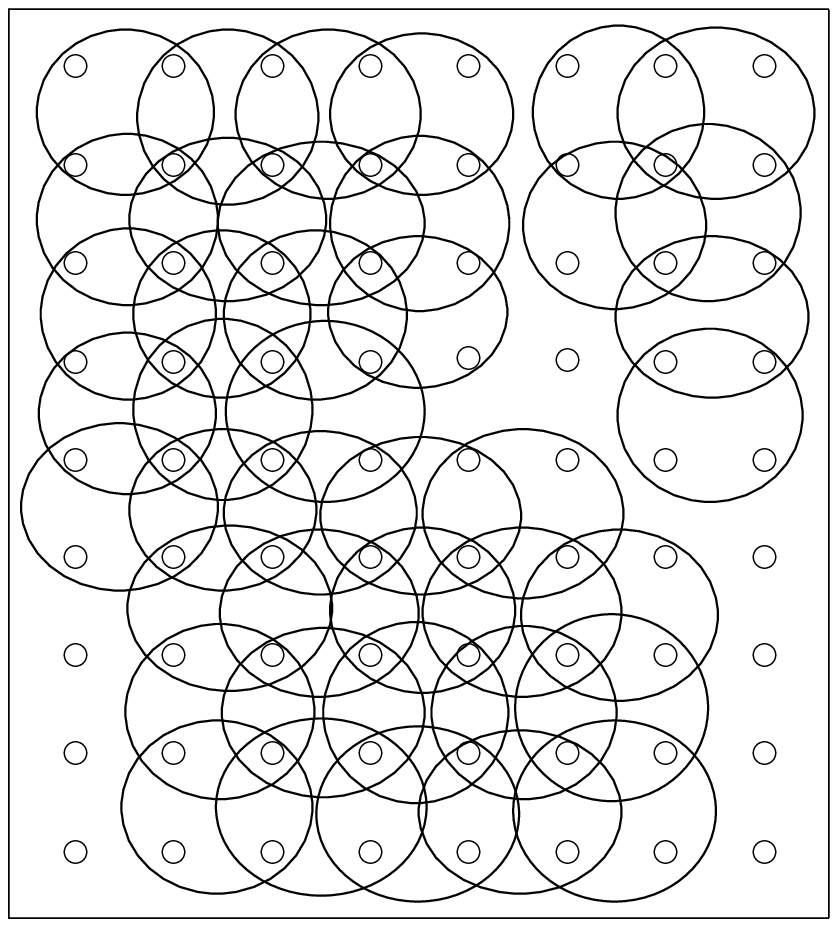}
\caption*{(f) $\delta^{\times}(X^{\bullet})$}
%\label{fig:figure2}
\end{minipage}
\hspace{2cm}
\begin{minipage}[b]{0.15\linewidth}
\centering
\includegraphics[scale=.5]{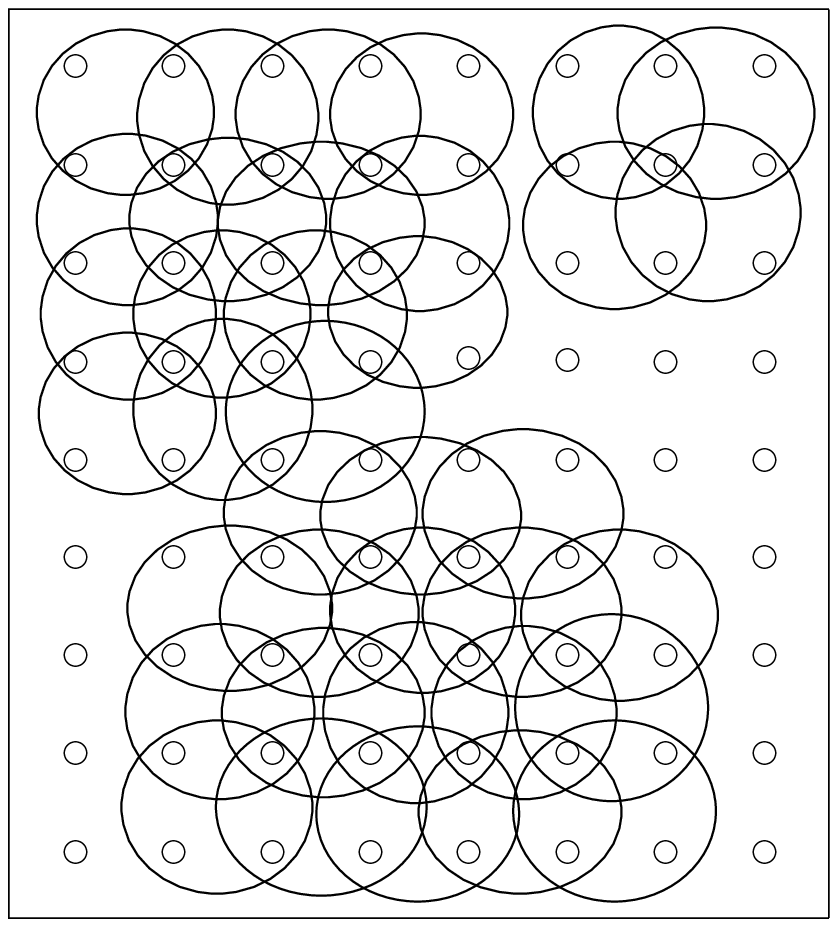}
\caption*{(g) $\delta(X^\bullet)$}
%\label{fig:figure2}
\end{minipage}
\hspace{2cm}
\begin{minipage}[b]{0.15\linewidth}
\centering
\includegraphics[scale=.5]{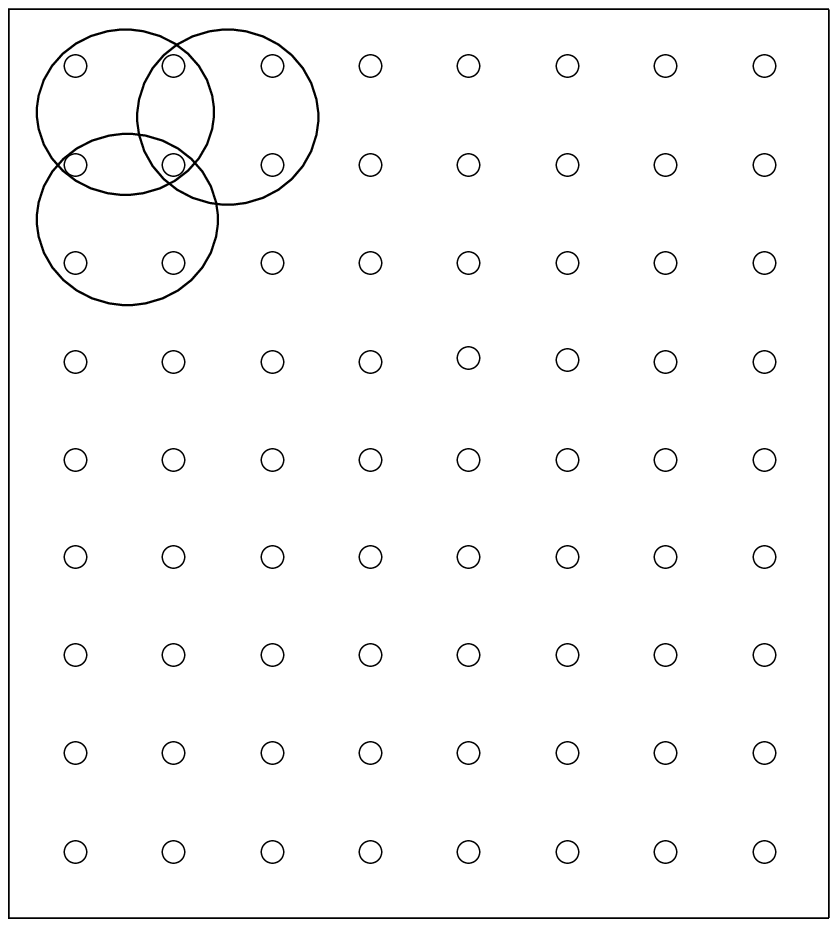}
\caption*{(h) $\epsilon(X^\bullet)$}
%\label{fig:figure2}
\end{minipage}
\hspace{2cm}
\begin{minipage}[b]{0.15\linewidth}
\centering
\includegraphics[scale=.5]{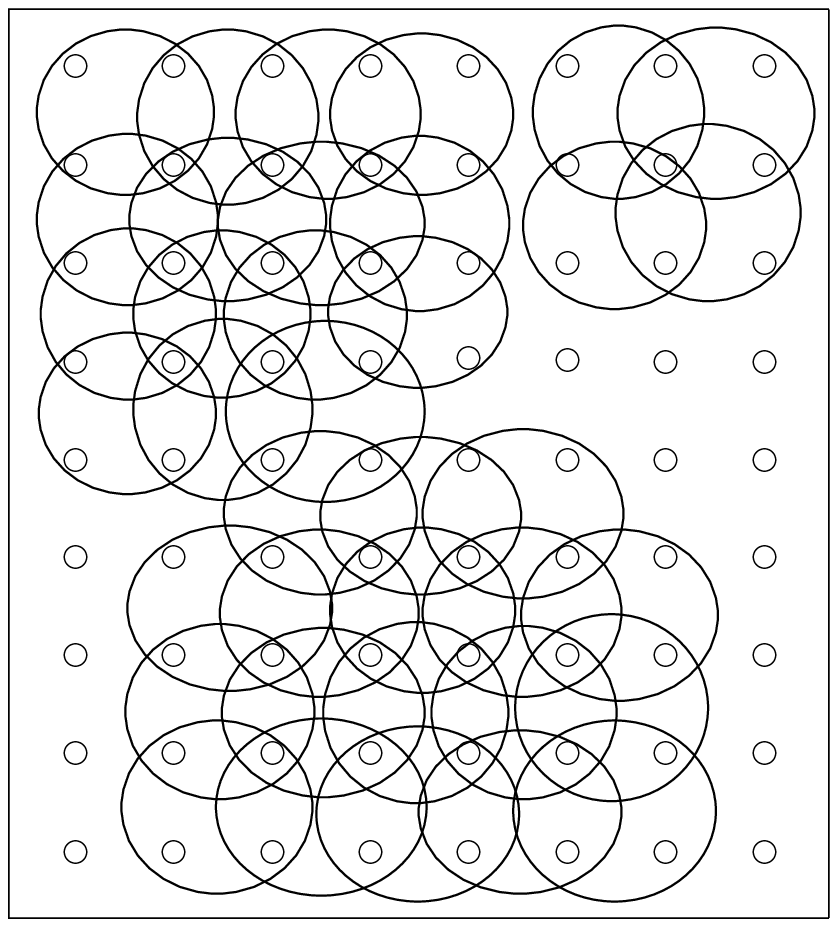}
\caption*{(i) $\bigtriangleup(X^\times)$}
%\label{fig:figure2}
\end{minipage}
\hspace{2cm}
\begin{minipage}[b]{0.15\linewidth}
\centering
\includegraphics[scale=.5]{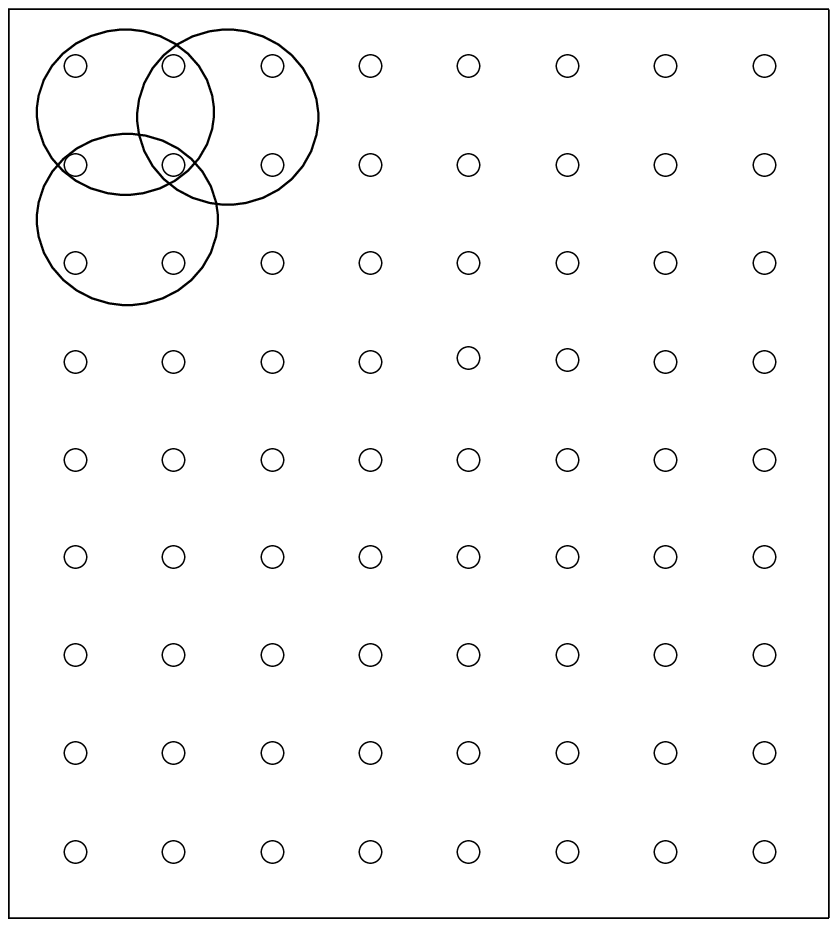}
\caption*{(j) $\varepsilon(X^\times)$}
%\label{fig:figure2}
\end{minipage}
\hspace{2cm}
\begin{minipage}[b]{0.15\linewidth}
\centering
\includegraphics[scale=.5]{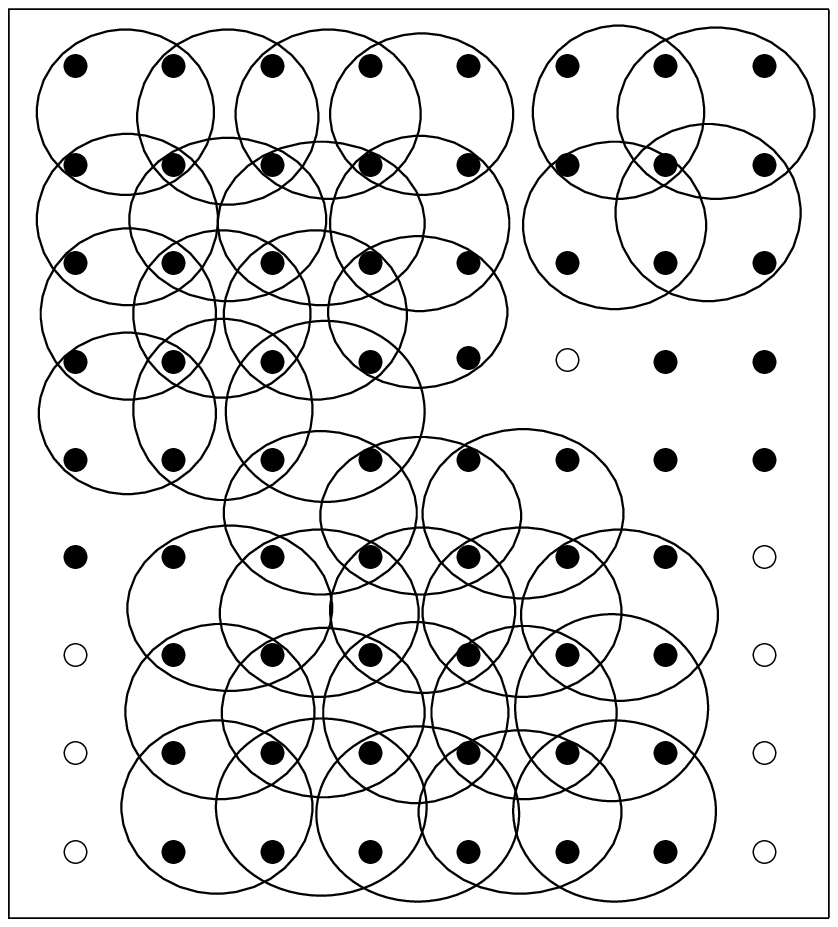}
\caption*{(k) $[\delta, \bigtriangleup](X)$}
%\label{fig:figure2}
\end{minipage}
\hspace{2cm}
\begin{minipage}[b]{0.15\linewidth}
\centering
\includegraphics[scale=.5]{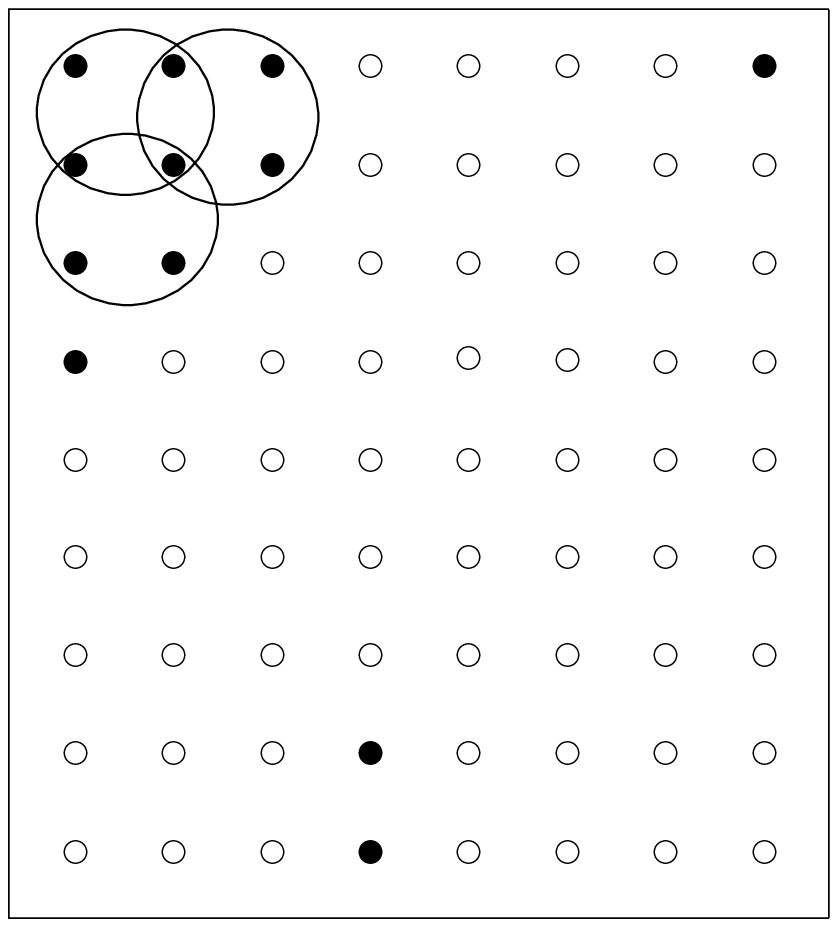}
\caption*{(l) $[\epsilon, \varepsilon](X)$}
%\label{fig:figure2}
\end{minipage}
\caption{Illustration of dilations and erosions}
\end{figure}
\end{center}

\begin{prop}\label{Property1}
For any \vv{X} $\subseteq$ \vv{H} and any \ee{X} $\subseteq $ \ee{H}, where \ee{X} $=(e_{j}), j \in J$ such that $J \subseteq I$
\begin{enumerate}
\item \vv{\delta}: \ee{H} $\rightarrow$ \vv{H} is such that $\delta^\bullet(X^\times)=\underset{j \in J}\cup v(e_{j})$;
\item \ee{\epsilon}: \vv{H} $\rightarrow$ \ee{H} is such that $\epsilon^\times(X^\bullet)=\lbrace e_{i},i \in I | v(e_{i}) \subseteq X^\bullet \rbrace$;
\item \vv{\epsilon}: \ee{H} $\rightarrow$ \vv{H} is such that $\epsilon^\bullet(X^\times)=\underset{j \notin J}\cap \overline{v(e_{j})}$;
\item \ee{\delta}: \vv{H} $\rightarrow$ \ee{H} is such that $\delta^\times(X^\bullet)=\lbrace e_{i},i \in I | v(e_{i}) \cap X^\bullet \neq \phi \rbrace$.
\end{enumerate}
\end{prop}
\begin{proof}
1. and 2. follows from the definition of \vv{\delta} and \ee{\epsilon}.
\begin{enumerate}
\item[3.] $H(\overline{X^\times})=(\underset{j \notin J}\cup v(e_{j}), (e_{j})_{j \notin J})$. Thus 
\begin{eqnarray*}
\epsilon^\bullet(X^\times) & = & \overline{\underset{j \notin J}\cup v(e_{j})}\\
%& = & \overline{\cup_{j \in I \setminus J}v(e_{j})}\\
%& = & \cap_{j \in I \setminus J}\overline{v(e_{j})} \text{~~~~~~~~~~~~~~~~~~~~~~~~~~~(By De Morgan's Law)}\\
& = & \underset{j \notin J} \cap \overline{v(e_{j})}\text{~~~~~~~~~~~~~~~~~~~~~~~~~~~~~~~~~~~~~~~~~~~~~~~~~~~~~~~~(By De Morgan's Law)}\\
\end{eqnarray*}
\item[4.] $\overline{\delta^\times(X^\bullet)}= \lbrace e_{i}, i \in I |v(e_{i}) \subseteq \overline{X^\bullet} \rbrace)$. Thus $\delta^\times(X^\bullet)=\lbrace e_{i}, i \in I |v(e_{i}) \cap X^\bullet \neq \phi \rbrace$.
\end{enumerate}
\end{proof}

Note that $\delta^\bullet(X^\times)= \lbrace x \in H^\bullet | \exists e_{j} \in X^\times \text{ such that } x \in v(e_{j}) \text{ for some } j \in J \rbrace$. This property states that $\delta^\bullet(X^\times)$ is the set of all vertices which belong to a hyperedge of $X^\times$. $\epsilon^\times(X^\bullet)$ is the set of all hyperedges whose vertices are composed of vertices of $X^\bullet$. $\epsilon^\bullet(X^\times)$ is the set of all vertices which do not belong to any edge of $\overline{X^\times}$, and $\delta^\times(X^\bullet)$ is the set of all hyperedges in $H^\times$ with atleast one vertex in $X^\bullet$. Therefore the previous property locally characterizes the operators defined in vertex-hyperedge correspondence. This property leads to simple linear time algorithms (with respect to $|H^\bullet|$ and $|H^\times|$) to compute $\delta^\bullet$, $\delta^\times$, $\epsilon^\bullet$ and $\epsilon^\times$.

\begin{prop}(dilation, erosion, adjunction, duality)\label{Property2}
\begin{enumerate}
\item Operators \ee{\epsilon} and \ee{\delta} (resp. \vv{\epsilon} and \vv{\delta}) are dual of each other.
\item Both (\ee{\epsilon}, \vv{\delta}) and (\vv{\epsilon}, \ee{\delta}) are adjunctions.
\item Operators \vv{\epsilon} and \ee{\epsilon} are erosions.
\item Operators \vv{\delta} and \ee{\delta} are dilations. 
\end{enumerate}
\end{prop}

\begin{proof} 
\begin{enumerate}
\item We will prove that $\overline{\delta^\times(\overline{X^\bullet})} = \epsilon^\times(X^\bullet)$ and $\overline{\delta^\bullet(\overline{X^\times})}= \epsilon^\bullet(X^\times)$
\begin{eqnarray*}
\delta^\times(\overline{X^\bullet}) & = & \lbrace e_{i},i \in I|v(e_{i}) \cap \overline{X^\bullet} \neq \phi \rbrace ~~~~~~~~~~~~~~~~~~~~~~~~~~~~~~~~~~~~~~~~~\text{(By property \ref{Property1} of $\delta^\times$)}\\
\overline{\delta^\times(\overline{X^\bullet})} & = & \lbrace e_{i},i \in I|v(e_{i}) \subseteq X^\bullet \rbrace \\
& = & \epsilon^\times(X^\bullet).
\end{eqnarray*}
Thus $\epsilon^\times$ and $\delta^\times$ are duals.
\begin{eqnarray*}
\delta^\bullet(X^\times) & = & \underset{j \in J}\cup v(e_{j}) \\
\delta^\bullet(\overline{X^\times}) & = & \underset{j \notin J}\cup v(e_{j}) \\
\overline{\delta^\bullet(\overline{X^\times})} & = & \overline{\underset{j \notin J}\cup v(e_{j})} \\ 
& = & \underset{j \notin J} \cap \overline{v(e_{j})} ~~~~~~~~~~~~~~~~~~~~~~~~~~~~~~~~~~~~~~~~~~~~~~~~~~~~~~~~\text{(By De Morgan's Law)}\\
& = & \epsilon^\bullet(X^\times).
\end{eqnarray*}
Therefore $\epsilon^\bullet$ and $\delta^\bullet$ are duals.
\item Suppose that $X^\times \subseteq \epsilon^\times(Y^\bullet)$. Then 
\begin{eqnarray*}
x \in \delta^\bullet(X^\times) & \Rightarrow & x \in \underset{j \in J}\cup v(e_{j})\\
& \Rightarrow & x \in v(e_{j}) \text{ for some } j \in J\\
& \Rightarrow & \exists ~e \in X^\times \text{ such that } x \in v(e)\\
& \Rightarrow & e \in \epsilon^\times(Y^\bullet)~~~~~~~~~~~~~~~~~~~~~~~~~~~~~~~~~~~~~~~~~~~~~~~~~~~~~~~~(\because X^\times \subseteq \epsilon^\times(Y^\bullet))\\
& \Rightarrow & e \in \lbrace e_{i}, i \in I |v(e_{i}) \subseteq Y^\bullet \rbrace\\
& \Rightarrow & v(e) \subseteq Y^\bullet \\
& \Rightarrow & x \in Y^\bullet ~~~~~~~~~~~~~~~~~~~~~~~~~~~~~~~~~~~~~~~~~~~~~~~~~~~~~~~~~~~~~~~~~~~~(\because x \in v(e))
\end{eqnarray*}
Therefore $\delta^\bullet(X^\times) \subseteq Y^\bullet$.\\
Conversly, if $\delta^\bullet(X^\times) \subseteq Y^\bullet$. Then 
\begin{eqnarray*}
e \in X^\times & \Rightarrow & v(e) \subseteq \delta^\bullet(X^\times)\\
& \Rightarrow & v(e) \subseteq Y^\bullet ~~~~~~~~~~~~~~~~~~~~~~~~~~~~~~~~~~~~~~~~~~~~~~~~~~~~~~~~~~~~~(\because \delta^\bullet(X^\times) \subseteq Y^\bullet)\\
& \Rightarrow & e \in \epsilon^\times(Y^\bullet)
\end{eqnarray*} 
Thus $X^\times \subseteq \epsilon^\times(Y^\bullet)$.
Therefore (\ee{\epsilon}, \vv{\delta}) is an adjunction.

\begin{eqnarray*}
\delta^\times(X^\bullet) \subseteq Y^\times & \Leftrightarrow & \overline{\epsilon^\times(\overline{X^\bullet})} \subseteq Y^\times ~~~~~~~~~~~~~~~~~~~~~~~~~~~~~~~~~~~~~~~~~~\text{(By duality of $\epsilon^\times$ and $\delta^\times$)}\\
& \Leftrightarrow & \overline{Y^\times} \subseteq \epsilon^\times(\overline{X^\bullet})\\
& \Leftrightarrow & \delta^\bullet(\overline{Y^\times}) \subseteq \overline{X^\bullet}~~~~~~~~~~~~~~~~~~~~~~~~~~~~~\text{(By adjunction property of (\ee{\epsilon}, \vv{\delta}))}\\
& \Leftrightarrow & X^\bullet \subseteq \overline{\delta^\bullet(\overline{Y^\times})}\\
& \Leftrightarrow & X^\bullet \subseteq \epsilon^\bullet(Y^\times)~~~~~~~~~~~~~~~~~~~~~~~~~~~~~~~~~~~~~~~~~~~~\text{(By duality of $\epsilon^\bullet$ and $\delta^\bullet$)}
\end{eqnarray*}
Therefore (\vv{\epsilon}, \ee{\delta}) is an adjunction.
\end{enumerate}
Properties 3. and 4. follows from the dilation / erosion property of adjunctions.
\end{proof}

\begin{defn}(\textbf{vertex dilation, vertex erosion}).
We define $\delta$ and $\epsilon$ that act on \X{H} by $\delta=\delta^\bullet \circ \delta^\times$ and $\epsilon=\epsilon^\bullet \circ \epsilon^\times$.
\end{defn}

\begin{prop}\label{Property3}
For any $X^\bullet \subseteq H^\bullet$.
\begin{enumerate}
\item $\delta(X^\bullet) = \lbrace x \in H^\bullet | ~\exists ~e_{i}, i \in I \text{ such that } x \in v(e_{i}) \text{ and } v(e_{i}) \cap X^\bullet \neq \phi \rbrace$.
\item $\epsilon(X^\bullet) = \lbrace x \in H^\bullet | ~\forall ~e_{i}, i \in I \text{ such that } x \in v(e_{i}), v(e_{i}) \subseteq X^\bullet \rbrace$.
\end{enumerate}
\end{prop}
\begin{proof}
\begin{enumerate}
\item 
\begin{eqnarray*}
\delta(X^\bullet) & = & \delta^\bullet ( \delta^\times(X^\bullet))\\
& = & \delta^\bullet[\lbrace e_{i}, i \in I | v(e_{i}) \cap X^\bullet \neq \phi \rbrace] ~~~~~~~~~~~~~~~~~~~~~~~~~~~~~~~\text{ (By property \ref{Property1} of $\delta^\times$})\\
& = & \underset{i \in I, v(e_{i}) \cap X^\bullet \neq \phi}\cup v(e_{i}) ~~~~~~~~~~~~~~~~~~~~~~~~~~~~~~~~~~~~~~~~~~~~~~\text{ (By property \ref{Property1} of $\delta^\bullet$})\\
& = & \lbrace x \in H^\bullet | ~\exists e_{i}, i \in I \text{ such that } x \in v(e_{i}) \text{ and } v(e_{i}) \cap X^\bullet \neq \phi \rbrace.
\end{eqnarray*}
\item 
\begin{eqnarray*}
\epsilon(X^\bullet) & = & \epsilon^\bullet(\epsilon^\times(X^\bullet))\\
& = & \epsilon^\bullet[\lbrace e_{i}, i \in I | v(e_{i}) \subseteq X^\bullet) \rbrace] ~~~~~~~~~~~~~~~~~~~~~~~~~~~~~~~~~~~~~~\text{ (By property \ref{Property1} of $\epsilon^\times$})\\
& = & \underset{i \in I, v(e_{i}) \nsubseteq X^\bullet}\cap \overline{v(e_{i})} ~~~~~~~~~~~~~~~~~~~~~~~~~~~~~~~~~~~~~~~~~~~~~~~~~~~~\text{ (By property \ref{Property1} of $\epsilon^\bullet$})\\
& = & \lbrace x \in X^\bullet | ~\forall e_{i} \in H^\times \text{ with } x \in v(e_{i}), v(e_{i}) \subseteq X^\bullet \rbrace
\end{eqnarray*}
\end{enumerate}
\end{proof}

\begin{defn}(\textbf{hyper-edge dilation, hyper-edge erosion})
We define $\bigtriangleup$ and $\varepsilon$ that act on \ee{\mathcal{H}} by $\bigtriangleup = \delta^\times \circ \delta^\bullet$ and $\varepsilon = \epsilon^\times \circ \epsilon^\bullet$.
\end{defn}

\begin{prop}\label{Property4}
For any $X^\times \subseteq H^\times$, $X^\times=(e_{j})_{j \in J}$.
\begin{enumerate}
\item $\bigtriangleup(X^\times) = \lbrace e_{i}, i \in I |~\exists e_{j}, j \in J \text{ such that } v(e_{i}) \cap v(e_{j}) \neq \phi \rbrace$.
\item $\varepsilon(X^\times) = \lbrace e_{j}, j \in J | v(e_{j}) \cap v(e_{i}) = \phi, \forall i \in I \setminus J \rbrace$.
\end{enumerate}
\end{prop}
\begin{proof}
\begin{enumerate}
\item 
\begin{eqnarray*}
\bigtriangleup(X^\times) & = & \delta^\times \circ \delta^\bullet(X^\times)\\
& = & \delta^\times [\underset{j \in J}\cup v(e_{j})]~~~~~~~~~~~~~~~~~~~~~~~~~~~~~~~~~~~~~~~~~~~~~~~~~~~~~~~~~~~~\text{ (By property \ref{Property1} of $\delta^\times$)}\\
& = & \lbrace e_{i}, i \in I | v(e_{i}) \cap [\underset{j \in J}\cup v(e_{j})] \neq \phi \rbrace.~~~~~~~~~~~~~~~~~~~~~~~~~~~~~~~~~~\text{ (By property \ref{Property1} of $\delta^\bullet$)}\\
& = & \lbrace e_{i}, i \in I | \exists e_{j}, j \in J \text{ such that } v(e_{i}) \cap v(e_{j}) \neq \phi \rbrace.
\end{eqnarray*}
\item
\begin{eqnarray*}
\varepsilon(X^\times) & = & \epsilon^\times \circ \epsilon^\bullet(X^\times)\\
& = & \epsilon^\times[\underset{i \in I \setminus J} \cap \overline{v(e_{i})}]~~~~~~~~~~~~~~~~~~~~~~~~~~~~~~~~~~~~~~~~~~~~~~~~~~~~~~~~~~~\text{ (By property \ref{Property1} of $\epsilon^\bullet$)}\\
& = & \lbrace e_{j}, j \in J | v(e_{j}) \subseteq [\underset{i \in I \setminus J}{\cap}\overline{v(e_{i})}]\\
& = & \lbrace e_{j}, j \in J | v(e_{j}) \subseteq [\overline{\underset{i \in I \setminus J}{\cup}v(e_{i})}] \rbrace ~~~~~~~~~~~~~~~~~~~~~~~~~~~~~~~~~~~\text{ (By De Morgan's Law)}\\
& = & \lbrace e_{j}, j \in J | v(e_{j}) \cap v(e_{i}) = \phi, \forall i \in I \setminus J \rbrace
\end{eqnarray*}
\end{enumerate}
\end{proof}

\begin{rem}
Being the compositions of respectively dilations and erosions, $\delta$ and $\epsilon$ are respectively a dilation and an erosion \cite{heijmans1997composing}. Moreover by composition of adjunctions and dual operators, $\delta$ and $\epsilon$ are dual and $(\epsilon, \delta)$ is an adjunction. In a similar manner $(\varepsilon, \bigtriangleup)$ is also an adjunction.
\end{rem}

\begin{defn}(\textbf{hypergraph dilation, hypergraph erosion})
We define the operators $[\delta, \bigtriangleup]$ and $[\epsilon, \varepsilon]$ by respectively $[\delta, \bigtriangleup](X) = (\delta(X^\bullet), \bigtriangleup(X^\times))$ and $[\epsilon, \varepsilon](X) = (\epsilon(X^\bullet), \varepsilon(X^\times))$, for any $X \in \mathcal{H}$.
\end{defn}
\begin{Theorem}
The operators $[\delta, \bigtriangleup]$ and $[\epsilon, \varepsilon]$ are respectively a dilation and an erosion acting on the lattice $(\mathcal{H}, \subseteq)$. 
\end{Theorem}
\begin{proof}
We will prove that for every $e \in \bigtriangleup(X^\times)$, $v(e) \subseteq \delta(X^\bullet)$. $e \in \bigtriangleup(X^\times)$ implies, there exists some $j \in J$ such that $v(e) \cap v(e_{j}) \neq \phi$. But $v(e_{j}) \subseteq X^\bullet$, since $j \in J$. Thus $v(e) \cap X^\bullet \neq \phi$. Therefore $v(e) \subseteq \underset{j \in J, v(e_{i}) \cap X^\bullet \neq \phi}\cup v(e_{i}) = \delta(X^\bullet)$. This implies $[\delta, \bigtriangleup](X) \in \mathcal{H}$.
\\If  $e \in \varepsilon(X^\times)$, then $v(e) \cap v(e_{i}) = \phi$ for every $i \in I \setminus J$, and so $v(e) \cap [\underset{i \in I \setminus J}\cup v(e_{i})] = \phi$.
\begin{eqnarray*}
v(e) & \subseteq & \overline{\underset{i \in I \setminus J}\cup v(e_{i})}\\
& = & \underset{i \in I \setminus J}\cap \overline{v(e_{i})}\\
& \subseteq & \underset{v(e_{i}) \nsubseteq X^\bullet}\cap \overline{v(e_{i})}  \text{~~~~~~~~~~~~~~~~~~~~~~~~~~~~~~~~~~~~~~~~~~~~~~~~~~~~~~~~~~~(Since $v(e_{i}) \subseteq X^\bullet, \forall i \in J)$}\\
& = & \varepsilon(X^\bullet)\text{~~~~~~~~~~~~~~~~~~~~~~~~~~~~~~~~~~~~~~~~~~~~~~~~~~~~~~~~~~~~~~~~~~~~~~~~~~(By Property \ref{Property3} of $\varepsilon$)}
\end{eqnarray*}
Therefore $[\epsilon, \varepsilon](X) \in \mathcal{H}$.
\end{proof}
\begin{Theorem}
$([\epsilon, \varepsilon], [\delta, \bigtriangleup])$ is an adjunction.
\end{Theorem}
\begin{proof}
Let $X$ and $Y$ are two hypergraphs in $\mathcal{H}$. The following statements are equivalent.
\begin{gather*}
[\delta, \bigtriangleup](X) \subseteq Y \\
\delta(X^\bullet) \subseteq Y^\bullet ~\text{and $\bigtriangleup(X^\times) \subseteq Y^\times$}\\
X^\bullet \subseteq \epsilon(Y^\bullet) ~\text{and } X^\times \subseteq \varepsilon(Y^\times) ~~~~~~~~~~~~~~~~~~~~~~~~~\text{(since $[\epsilon, \varepsilon]$ and $[\delta, \bigtriangleup]$ are adjunctions on $\mathcal{H}$)}\\
X \subseteq [\epsilon, \varepsilon](Y)
\end{gather*}
Thus the pair $([\epsilon, \varepsilon], [\delta, \bigtriangleup])$ is an adjunction, which implies that $[\epsilon, \varepsilon]$ is an erosion and $[\delta, \bigtriangleup]$ is a dilation.
\end{proof}
\section{Filters}
In mathematical morphology, a $filter$ \cite{cousty2013morphological}, \cite{najman2013mathematical} is an operator $\alpha$ acting on a lattice $\mathcal{L}$, which is increasing ($i.e. \forall X, Y \in \mathcal{L}, X \leq Y \implies \alpha(X) \leq \alpha(Y) $) and idempotent ($i.e. \forall X \in \mathcal{L}, \alpha( \alpha (X)) = \alpha(X)$). A filter on $\mathcal{L}$ which is extensive ($i.e. \forall X \in \mathcal{L}, X \leq \alpha(X)$) is called a $closing$ on $\mathcal{L}$ and a filter on $\mathcal{L}$ which is anti-extensive ($i.e. \forall X \in \mathcal{L}, \alpha(X) \leq X$) is called an $opening$. If $(\alpha, \beta)$ is an adjunction then $\alpha$ is an erosion, $\beta$ is a dilation, $\beta \circ \alpha$ is an opening and $\alpha \circ \beta$ is a closing on $\mathcal{L}$.
\begin{defn}(\textbf{opening, closing}).
\begin{itemize}
\item[1.]
We define $\gamma_{1}$ and $\phi_{1}$, that act on $\mathcal{H^\bullet}$, by $\gamma_{1}=\delta \circ \epsilon$ and $\phi_{1}=\epsilon \circ \delta$.
\item[2.]
We define $\Gamma_{1}$ and $\Phi_{1}$, that act on $\mathcal{H^\times}$, by $\Gamma_{1}=\Delta \circ \varepsilon$ and $\phi_{1}=\varepsilon \circ \Delta$.
\item[3.]
We define $[\gamma,\Gamma]_{1}$ and $[\phi,\Phi]_{1}$, that act on $\mathcal{H}$ by respectively $[\gamma,\Gamma]_1(X)=(\gamma_{1}(X^\bullet),\Gamma_{1}(X^\times))$ and $[\phi,\Phi]_{1}(X)=(\phi_{1}(X^\bullet),\Phi_{1}(X^\times))$ for any $X \in \mathcal{H}$.
\end{itemize}
\end{defn}
Since $(\epsilon, \delta)$ and $(\varepsilon, \Delta)$ are adjunctions, $\gamma_{1}$, $\Gamma_{1}$ are openings and $\phi_{1}$, $\Phi_{1}$ are closings on the respective lattices. Now we will prove that  $[\gamma,\Gamma]_{1}$ and $[\phi,\Phi]_{1}$ are respectively an opening and a closing on $\mathcal{H}$.
\begin{Proposition}
The following statements are true.
\begin{itemize}
\item[1.] $[\gamma,\Gamma]_{1} = [\delta, \Delta] \circ [\epsilon, \varepsilon]$
\item[2.] $[\phi,\Phi]_{1} = [\epsilon, \varepsilon] \circ [\delta, \Delta]$
\end{itemize}
\end{Proposition}
\begin{proof}
Let $X$ be any hypergraph in $\mathcal{H}$. Then 
\begin{eqnarray*}
[\gamma,\Gamma]_{1}(X)& =& (\gamma_{1}(X^\bullet), \Gamma_{1}(X^\times))\\
&=& ((\delta \circ \epsilon)(X^\bullet), (\Delta \circ \varepsilon)(X^\times)\\
&=& (\delta(X^\bullet), \Delta(X^\times)) \circ (\epsilon(X^\bullet), \varepsilon(X^\times))\\
&=& [\delta, \Delta] \circ [\epsilon, \varepsilon](X)
\end{eqnarray*}
This proves 1. A similar line of arguments will prove 2.
\end{proof}
$([\epsilon, \varepsilon], [\delta, \Delta])$ is an adjunction on $\mathcal{H}$ implies that $[\delta, \Delta] \circ [\epsilon, \varepsilon]$ is an opening and $[\epsilon, \varepsilon] \circ [\delta, \Delta]$ is a closing on  $\mathcal{H}$.
\begin{defn}(\textbf{half-opening, half-closing}).
\begin{itemize}
\item[1.]
We define $\gamma_{1/2}$ and $\phi_{1/2}$, that act on $\mathcal{H^\bullet}$, by $\gamma_{1/2}=\delta^\bullet \circ \epsilon^\times$ and $\phi_{1/2}=\epsilon^\bullet \circ \delta^\times$.
\item[2.]
We define $\Gamma_{1/2}$ and $\Phi_{1/2}$, that act on $\mathcal{H^\times}$, by $\Gamma_{1/2}=\delta^\times \circ \epsilon^\bullet$ and $\phi_{1/2}=\epsilon^\times \circ \delta^\bullet$.
\item[3.]
We define $[\gamma,\Gamma]_{1/2}$ and $[\phi,\Phi]_{1/2}$, that act on $\mathcal{H}$ by respectively $[\gamma,\Gamma]_{1/2}(X)=(\gamma_{1/2}(X^\bullet),\Gamma_{1/2}(X^\times))$ and $[\phi,\Phi]_{1/2}(X)=(\phi_{1/2}(X^\bullet),\Phi_{1/2}(X^\times))$ for any $X \in \mathcal{H}$.
\end{itemize}
\end{defn}
\begin{prop}
Let $X^\bullet \subseteq H^\bullet$ and $X^\times \subseteq H^\times$. The following properties are true.
\begin{enumerate}
\item[1.] $\gamma_{1/2}(X^\bullet)=\underset{i \in I, v(e_{i}) \subseteq X^\bullet} \cup v(e_{i})$
%\item[1.] $\gamma_{1/2}(X^\bullet)=\lbrace x \in X^\bullet| \exists e_{i}, i \in I~ \text{with} ~ x \in v(e_{i})~ \text{and}~ v(e_{i}) \subseteq X^\bullet \rbrace$.
%\item[2.] $\gamma_{1/2}(X^\bullet)=X^\bullet \setminus \lbrace x \in X^\bullet| \forall e_{i}, i \in I~ \text{with} ~ x \in v(e_{i})~ \text{and}~ v(e_{i}) \nsubseteq X^\bullet \rbrace$.
\item[2.] $\phi_{1/2}(X^\bullet)= \lbrace x \in H^\bullet| \forall e_{i}, i \in I~ \text{with} ~ x \in v(e_{i})~ \text{and}~ v(e_{i}) \cap X^\bullet \neq \phi \rbrace$.
\item[3.] $\Gamma_{1/2}(X^\times)= \lbrace e_{i}, i \in I | \exists~ x \in v(e_{i}) ~\text{with}~  \lbrace e_{i} \in H^\times ~\text{with}~ x \in v(e_{i}) \rbrace \subseteq X^\times \rbrace$.
%\item[5.] $\phi_{1/2}(X^\times)= \lbrace e_{i}, i \in I | v(e_{i}) \subseteq \delta^\bullet(X^\times) \rbrace$.
\item[4.] $\phi_{1/2}(X^\times)= \lbrace e_{i}, i \in I | v(e_{i}) \subseteq \underset{j \in J}\cup v(e_{j}) \rbrace$.
\end{enumerate}
\end{prop}

\begin{proof}
\begin{enumerate}
\item[1.]$\gamma_{1/2}(X^\bullet) = \delta^\bullet \circ \epsilon^\times(X^\bullet)$
\begin{eqnarray*}
%\gamma_{1/2}(X^\bullet)&=&\delta^\bullet \circ \epsilon^\times(X^\bullet)\\
&=& \delta^\bullet \circ \lbrace e_{i}, i \in I|v(e_{i}) \subseteq X^\bullet \rbrace \text{ (By property of $\epsilon^\times)$}\\
&=& \underset{i \in I, v(e_{i}) \subseteq X^\bullet} \cup v(e_{i}) \text{ (By property of $\delta^\bullet)$}
\end{eqnarray*}
\item[2.] $\phi_{1/2}(X^\bullet)= \epsilon^\bullet \circ \delta^\times(X^\bullet)$
\begin{eqnarray*}
&=& \epsilon^\bullet \circ \lbrace e_{i}, i \in I| v(e_{i}) \cap X^\bullet \neq \phi \rbrace\\
&=& \epsilon^\bullet \circ \lbrace e_{k}, k \in K \rbrace ;\text{ where} K \subseteq I \text{ is some index set and $e_{k}$ is such that $v(e_{k}) \cap X^\bullet \neq \phi$}\\
& &\text{(By property of $\delta^\times(X^\bullet)$)}\\
&=& \underset{k \notin K} \cap \overline{v(e_{k})}
\end{eqnarray*}
\item[3.] $\Gamma_{1/2}(X^\times)= \delta^\times \circ \epsilon^\bullet(X^\times)$
\begin{eqnarray*}
&=& \lbrace e_{i}, i \in I | v(e_{i}) \cap \epsilon^\bullet(X^\times) \neq \phi \rbrace \text{ (By property of $\delta^\times$)}\\
&=& \text{ set of all edges in $H^\times$ which do not belong to any edge of $\overline{X^\times}$}\\
&=& \lbrace e_{i}, i \in I | \exists x \in v(e_{i}) \text{ with } \lbrace e_{i} \in H^\times \text{ with } x \in v(e_{i}) \rbrace \subseteq X^\times \rbrace
\end{eqnarray*}
\item[4.] $\phi_{1/2}(X^\times)= \epsilon^\times \circ \delta^\bullet(X^\times)$
\begin{eqnarray*}
&=& \epsilon^\times \circ \underset{j \in J} \cup v(e_{j})\\
&=& \lbrace e_{i}, i \in I | v(e_{i}) \subseteq  \underset{j \in J} \cup v(e_{j}) \rbrace
\end{eqnarray*}
\end{enumerate}
\end{proof}

\begin{center}
\begin{figure}[htp!]
\begin{minipage}[b]{0.15\linewidth}
\centering
\centering
%\includegraphics[scale=.5]{H.eps}
%\hspace{2cm}
%\begin{minipage}[b]{0.15\linewidth}
%\centering
\includegraphics[scale=.5]{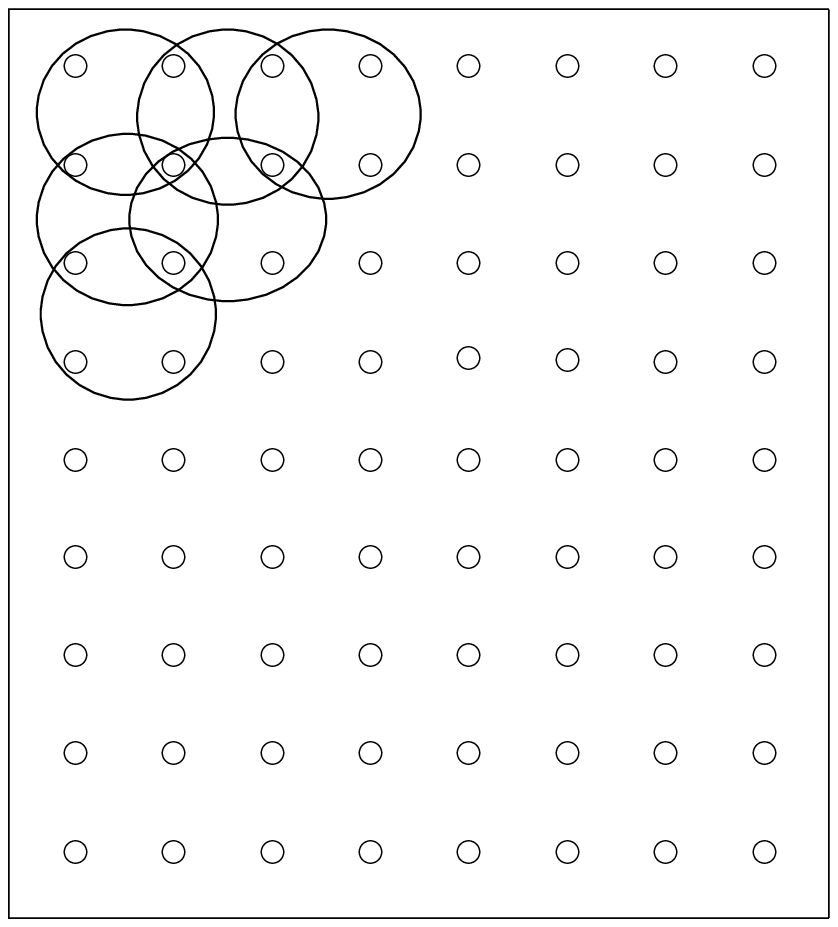}
\caption*{(m) $\gamma 1$}
%\label{fig:figure2}
\end{minipage}
\hspace{2cm}
\begin{minipage}[b]{0.15\linewidth}
\centering
\includegraphics[scale=.5]{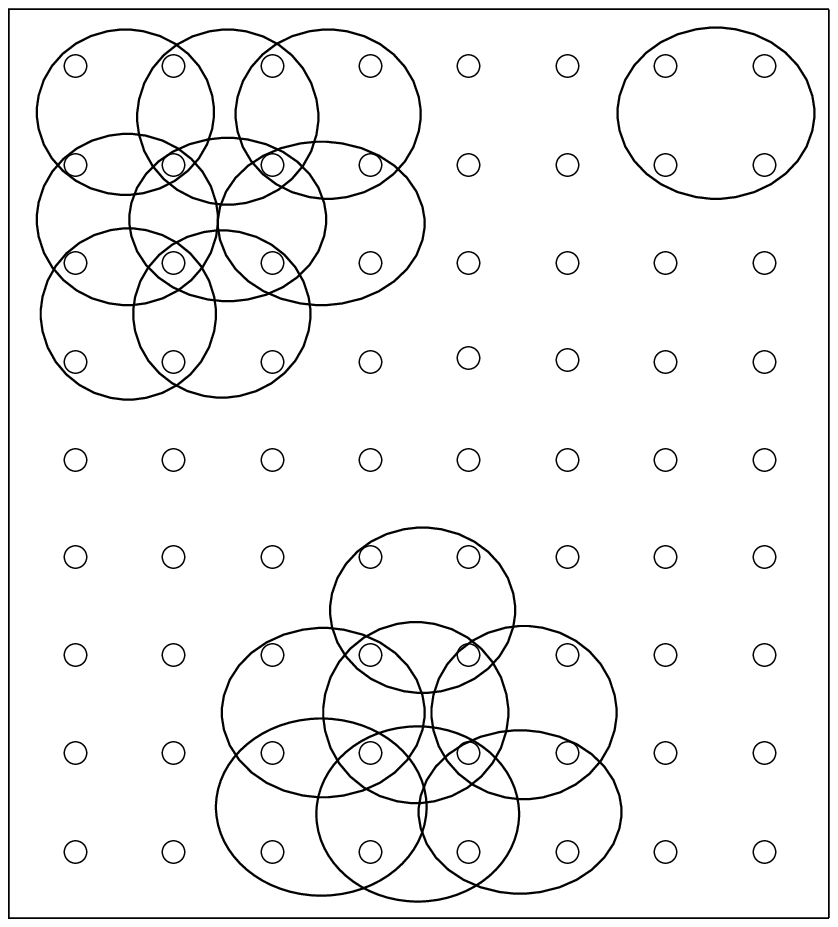}
\caption*{(n) $\phi 1 $}
%\label{fig:figure2}
\end{minipage}
\hspace{2cm}
\begin{minipage}[b]{0.15\linewidth}
\centering
\includegraphics[scale=.5]{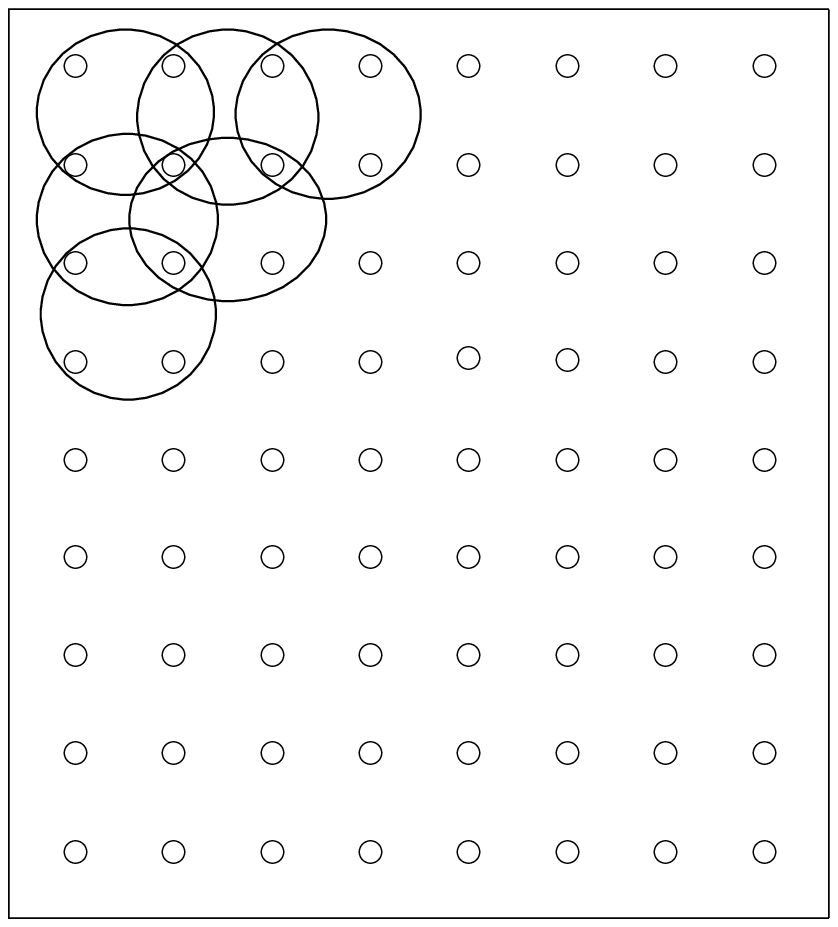}
\caption*{(o) $\Gamma 1$}
%\label{fig:figure2}
\end{minipage}
\hspace{2cm}
\begin{minipage}[b]{0.15\linewidth}
\centering
\includegraphics[scale=.5]{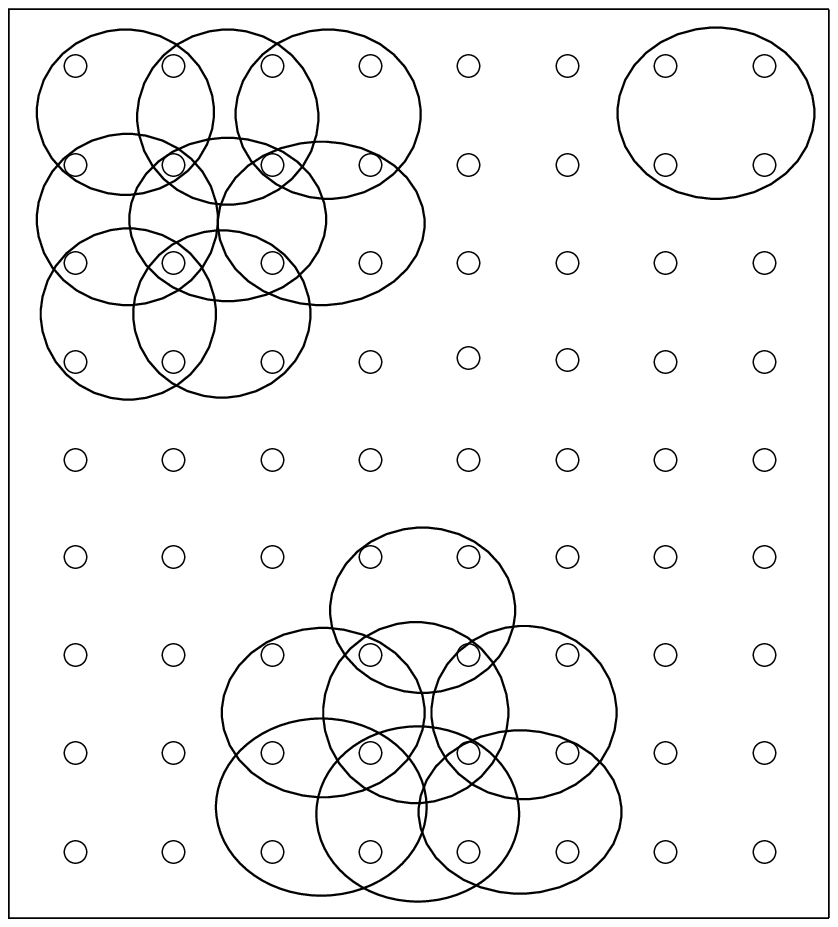}
\caption*{(p) $\varphi 1$}
%\label{fig:figure2}
\end{minipage}
\hspace{2cm}
\begin{minipage}[b]{0.15\linewidth}
\centering
\includegraphics[scale=.5]{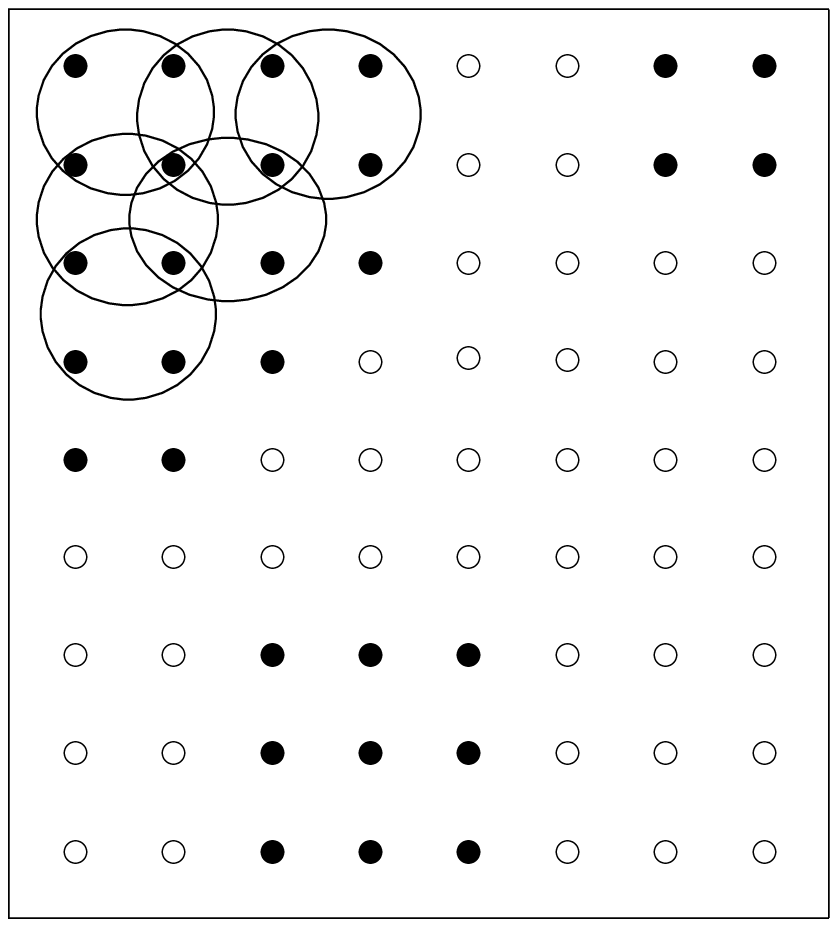}
\caption*{(q) $[\gamma, \Gamma]1$}
%\label{fig:figure2}
\end{minipage}
\hspace{2cm}
\begin{minipage}[b]{0.15\linewidth}
\centering
\includegraphics[scale=.5]{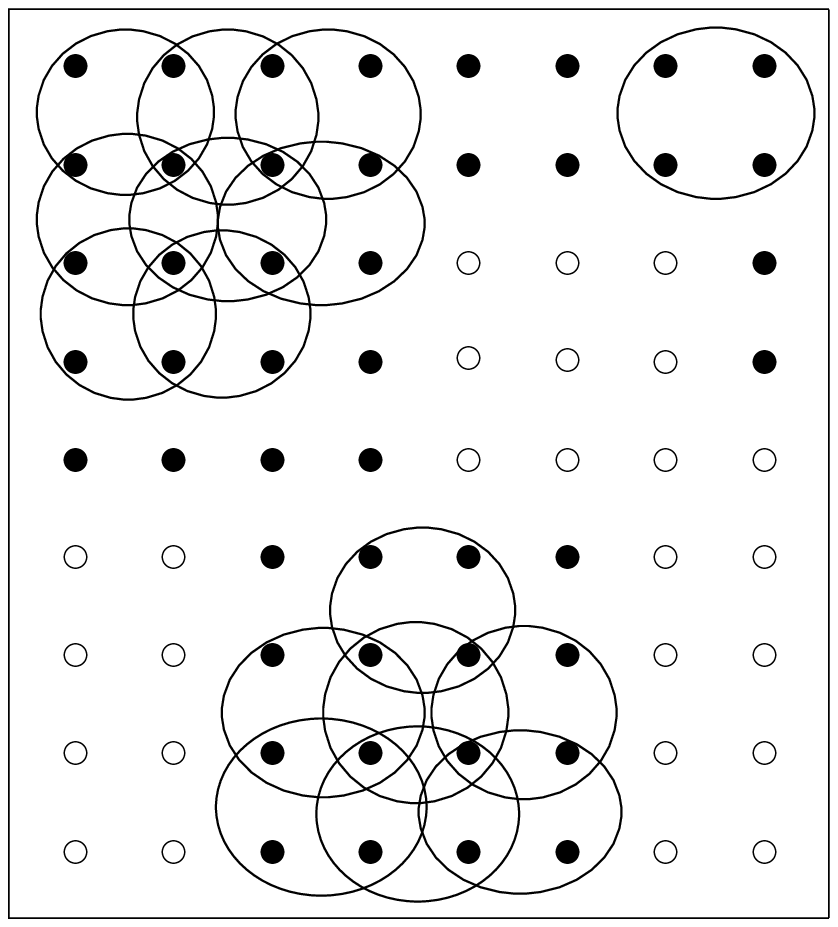}
\caption*{(r) $[\phi, \varphi]1$}
%\label{fig:figure2}
\end{minipage}
\hspace{2cm}
\begin{minipage}[b]{0.15\linewidth}
\centering
\includegraphics[scale=.5]{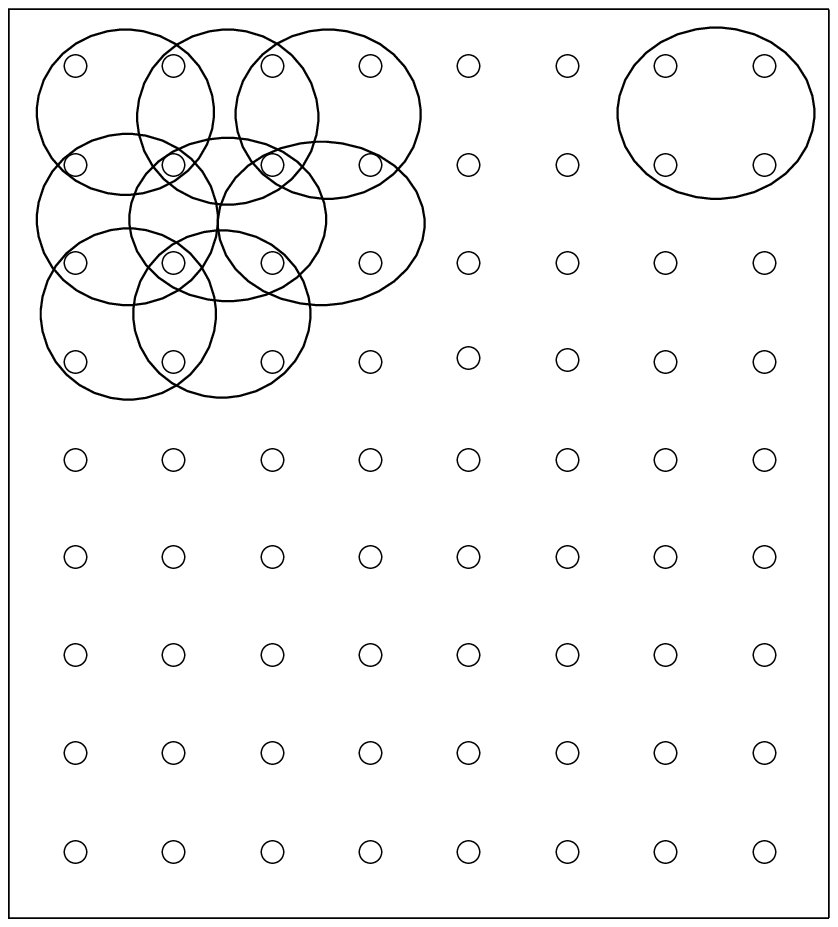}
\caption*{(s) $\gamma 1/2$}
%\label{fig:figure2}
\end{minipage}
\hspace{2cm}
\begin{minipage}[b]{0.15\linewidth}
\centering
\includegraphics[scale=.5]{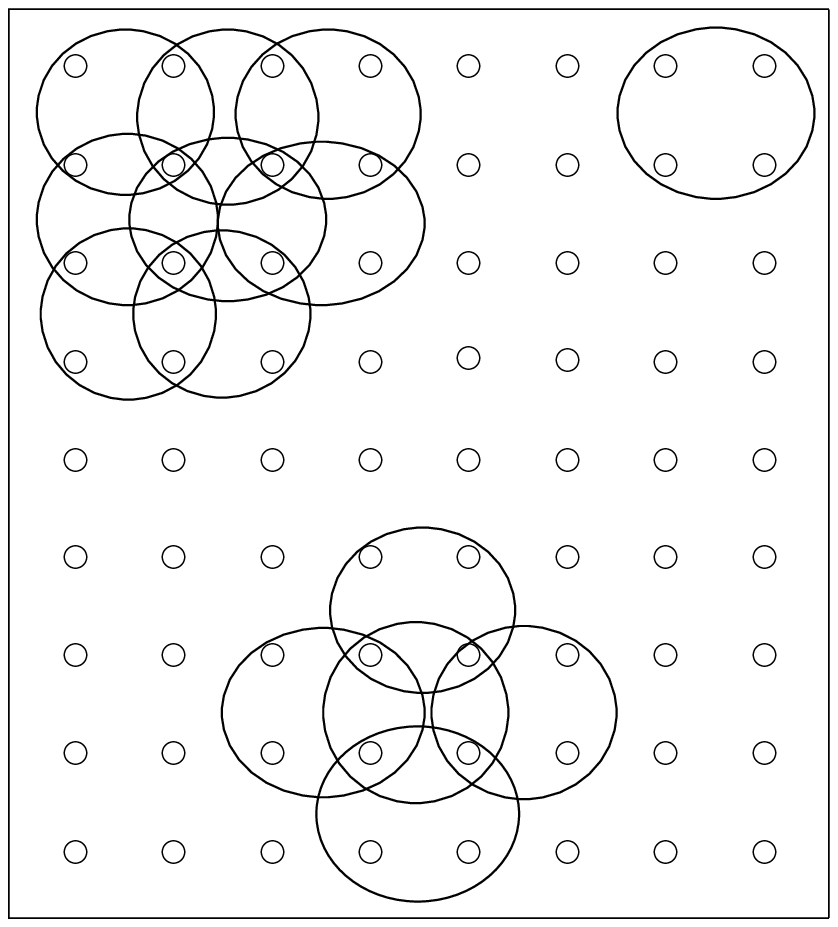}
\caption*{(t) $\phi 1/2$}
%\label{fig:figure2}
\end{minipage}
\hspace{2cm}
\begin{minipage}[b]{0.15\linewidth}
\centering
\includegraphics[scale=.5]{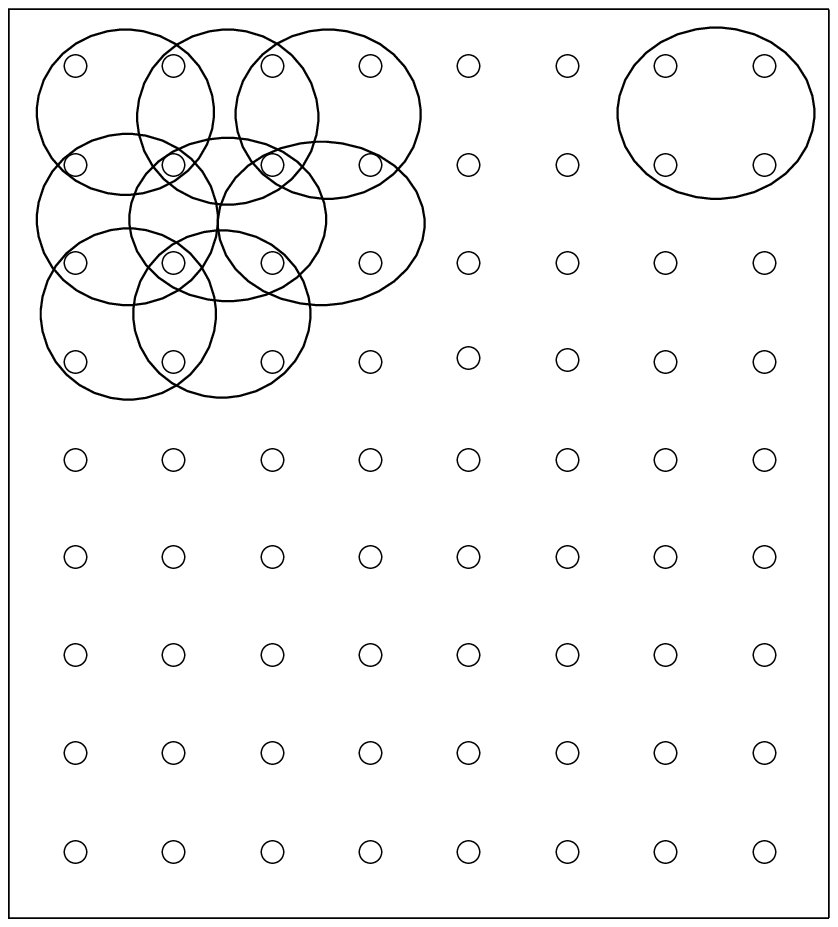}
\caption*{(u) $\Gamma 1/2$}
%\label{fig:figure2}
\end{minipage}
\hspace{2cm}
\begin{minipage}[b]{0.15\linewidth}
\centering
\includegraphics[scale=.5]{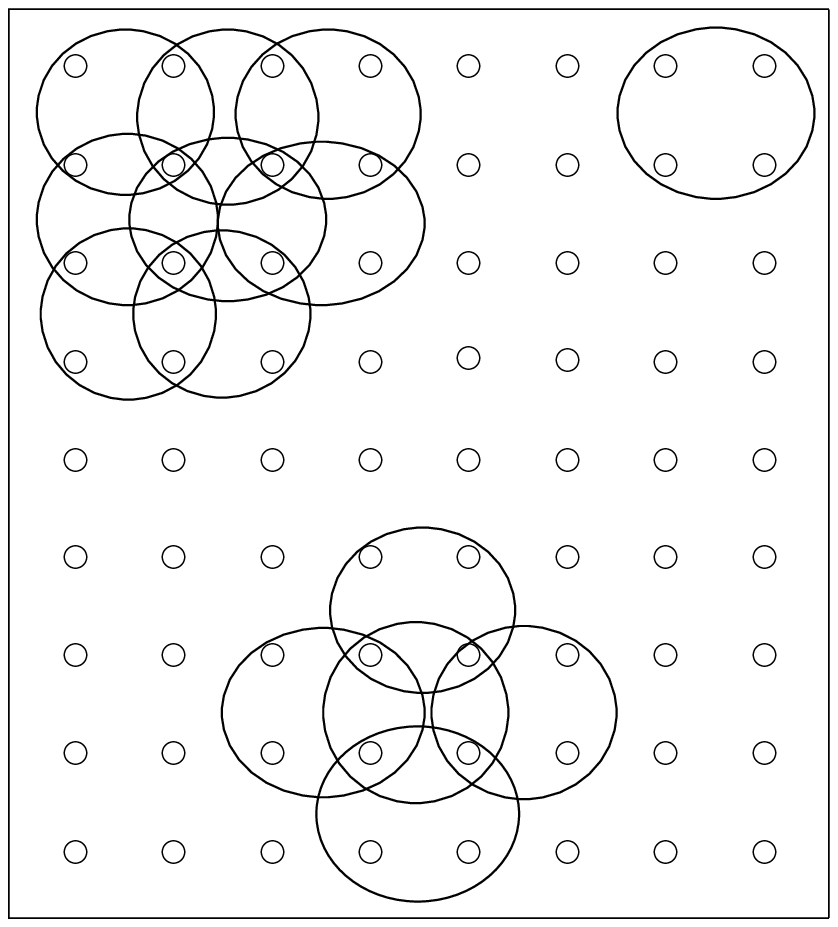}
\caption*{(v) $\varphi 1/2$}
%\label{fig:figure2}
\end{minipage}
\hspace{2cm}
\begin{minipage}[b]{0.15\linewidth}
\centering
\includegraphics[scale=.5]{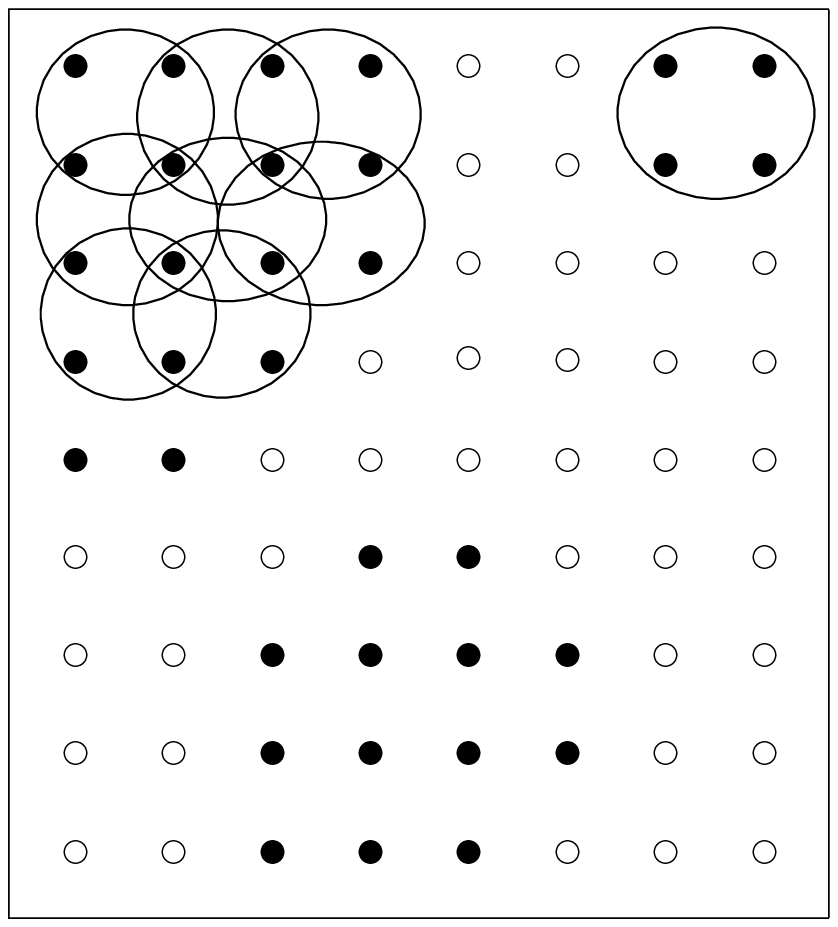}
\caption*{(w) $[\gamma, \Gamma]1/2$}
%\label{fig:figure2}
\end{minipage}
\hspace{2cm}
\begin{minipage}[b]{0.15\linewidth}
\centering
\includegraphics[scale=.5]{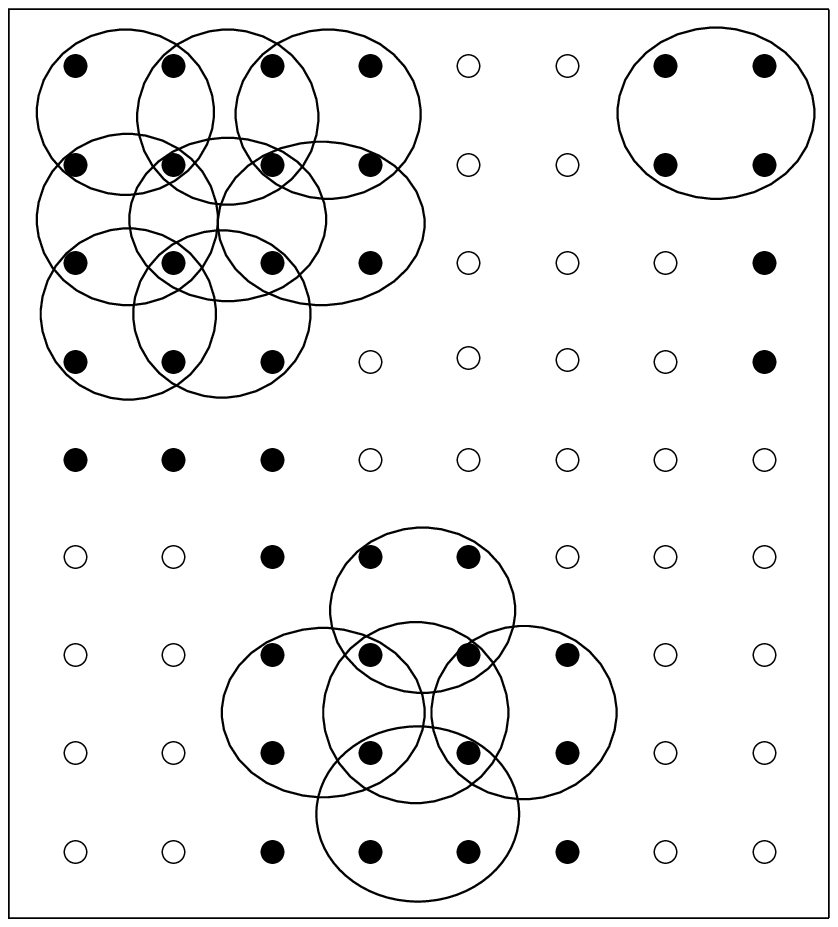}
\caption*{(x) $[\phi, \varphi]1/2$}
%\label{fig:figure2}
\end{minipage}

\caption{Illustration of openings and closings}
\end{figure}
\end{center}
\begin{rem}
The following statements are true about $\gamma_{1/2}$.
\begin{enumerate}
\item[1.] $\gamma_{1/2}(X^\bullet)=\lbrace x \in X^\bullet| \exists e_{i}, i \in I~ \text{with} ~ x \in v(e_{i})~ \text{and}~ v(e_{i}) \subseteq X^\bullet \rbrace$.
\item[2.] $\gamma_{1/2}(X^\bullet)=X^\bullet \setminus \lbrace x \in X^\bullet| \forall e_{i}, i \in I~ \text{with} ~ x \in v(e_{i})~ \text{and}~ v(e_{i}) \nsubseteq X^\bullet \rbrace$.
%\item[3.] $\phi_{1/2}(X^\times)= \lbrace e_{i}, i \in I | v(e_{i}) \subseteq \delta^\bullet(X^\times) \rbrace$.
\end{enumerate}
\end{rem}

Now we will enumerate some properties of these operators on the respective lattices based on the partial order relations defined on them.
\begin{prop} \label{Property6}
Let $X^\bullet \subseteq H^\bullet$ and $X^\times \subseteq H^\times$. The following properties hold true.
\begin{enumerate}
\item[1.] $\gamma_{1}(X^\bullet) \subseteq \gamma_{1/2}(X^\bullet) \subseteq X^\bullet \subseteq \phi_{1/2}(X^\bullet) \subseteq \phi_{1}(X^\bullet)$.
\item[2.] $\Gamma_{1}(X^\times) \subseteq \Gamma_{1/2}(X^\times) \subseteq X^\times \subseteq \Phi_{1/2}(X^\times) \subseteq \Phi_{1}(X^\times)$.
\item[3.] $[\gamma, \Gamma]_{1}(X) \subseteq [\gamma, \Gamma]_{1/2}(X) \subseteq X \subseteq [\phi, \Phi]_{1/2}(X) \subseteq [\phi, \Phi]_{1}(X)$.
\end{enumerate}
\end{prop}
\begin{proof}
\begin{enumerate}
\item[1.] 
\begin{eqnarray*}
\gamma_{1}(X^\bullet) &=& \delta \circ \epsilon (X^\bullet)\\
&=& \lbrace x \in H^\bullet | \exists e_{i}, i \in I \text{ such that } x \in v(e_{i}) \text{ and } v(e_{i}) \cap \epsilon(X^\bullet) \neq \phi \rbrace \text{~(By property of $\delta$ )}\\
& \subseteq & \lbrace x \in H^\bullet | \exists e_{i}, i \in I \text{ such that } x \in v(e_{i}) \text{ and } v(e_{i}) \subseteq X^\bullet \rbrace\\
&=& \underset{v(e_{i}) \subseteq X^\bullet} \cup v(e_{i})\\
&=& \gamma_{1/2}(X^\bullet).\\
\end{eqnarray*}
Now $ \gamma_{1/2}(X^\bullet) = \underset{v(e_{i}) \subseteq X^\bullet} \cup v(e_{i}) \subseteq X^\bullet$. Also $ X^\bullet \subseteq \phi_{1/2}(X^\bullet)$, because $ \phi_{1/2} $ is a closing on $H^\bullet$.
\begin{eqnarray*}
\phi_{1/2}(X^\bullet) &=& \epsilon^\bullet \circ \delta^\times(X^\bullet)\\
&=& \epsilon^\bullet \circ I \circ \delta^\times(X^\bullet), \text{~(where $I$ is the identity on $H^\times$)}\\
& \subseteq & \epsilon^\bullet \circ (\epsilon^\times \circ \delta^\bullet) \circ \delta^\times(X^\bullet) \text{ ($ \because  I(X^\times) = X^\times \subseteq \epsilon^\times \circ \delta^\bullet(X^\times)$ and $\epsilon^\times \circ \delta^\bullet$ is a closing)}\\
& = & (\epsilon^\bullet \circ \epsilon^\times) \circ (\delta^\bullet \circ \delta^\times)(X^\bullet)\\
&=& \epsilon \circ \delta(X^\bullet)\\
&=& \phi_{1}(X^\bullet).
\end{eqnarray*}
This proves 1.
\end{enumerate}
Properties 2 and 3 can be proved in a similar manner.
\end{proof}

\begin{prop}{\textbf{[hypergraph opening, hypergraph closing]}}
\begin{enumerate}
\item[1.] The operators $\gamma_{1/2}$ and $\gamma_{1}$ (resp. $\Gamma_{1/2}$ and $\Gamma_{1}$) are openings on $\mathcal{H^\bullet}$ (resp. $\mathcal{H^\times}$) and $\phi_{1/2}$ and $\phi_{1}$ (resp. $\Phi_{1/2}$ and $\Phi_{1}$) are closings on $\mathcal{H^\bullet}$ (resp. $\mathcal{H^\times}$).
\item[2.] The family $\mathcal{H}$ is closed under $[\gamma, \Gamma]_{1/2}$, $[\phi, \Phi]_{1/2}$, $[\gamma, \Gamma]_{1}$ and $[\phi, \Phi]_{1}$.
\item[3.] $[\gamma, \Gamma]_{1/2}$ and $[\gamma, \Gamma]_{1}$ are openings on $\mathcal{H}$ and $[\phi, \Phi]_{1/2}$ and $[\phi, \Phi]_{1}$ are closings on $\mathcal{H}$.
\end{enumerate}
\end{prop}
\begin{proof}
\begin{enumerate}
\item[1.] If $(\alpha, \beta)$ is an adjunction on a lattice $\mathcal{L}$ then  $\beta \circ \alpha$ is an opening and $\alpha \circ \beta$ is a closing on $\mathcal{L}$. The result follows from the fact that $(\epsilon^\bullet, \delta^\times)$, $(\epsilon^\times, \delta^\bullet)$, $(\epsilon, \delta)$ and $(\varepsilon, \Delta)$ are adjunctions.
\item[2.] We will prove that $[\gamma, \Gamma]_{1/2}(X) \in \mathcal{H}$, whenever $X \in \mathcal{H}$. Since $[\gamma, \Gamma]_{1/2}(X)= (\gamma_{1/2}(X^\bullet), \Gamma_{1/2}(X^\times))$, it is enough to show that if $e \in \Gamma_{1/2}(X^\times)$, then $v(e) \subseteq \gamma_{1/2}(X^\bullet)$.
\begin{eqnarray*}
e \in \Gamma_{1/2}(X^\times) & \Rightarrow & e \in \delta^\times \circ \epsilon^\bullet(X^\times)\\
& \Rightarrow & e \in X^\times \text{ (since $\delta^\times \circ \epsilon^\bullet$ is an opening $\delta^\times \circ \epsilon^\bullet(X^\times) \subseteq X^\times$)}\\
& \Rightarrow & e \in \epsilon^\times(X^\bullet) \text{ (since $X^\times \subseteq \epsilon^\times(X^\bullet)$)}
\end{eqnarray*}
Since $e_{j} \in X^\times \Rightarrow v(e_{j}) \subseteq \delta^\bullet(X^\times)$, we have $v(e) \subseteq \delta^\bullet \circ \epsilon^\times(X^\bullet)$. Thus $v(e) \subseteq \gamma_{1/2}(X^\bullet)$. Hence if $X$ is a hypergraph, then $[\gamma, \Gamma]_{1/2}(X)$ is a hypergraph and so $\mathcal{H}$ is closed under $[\gamma, \Gamma]_{1/2}$.\\
$\mathcal{H}$ is closed under $[\gamma, \Gamma]_{1}$ and $[\phi, \Phi]_{1}$ due to proposition 1.

Now we will prove that $\mathcal{H}$ is closed under $[\phi, \Phi]_{1/2}$. It is enough to show that if $X \in \mathcal{H}$, then $[\phi, \Phi]_{1/2}(X) \in \mathcal{H}$. We know that $[\phi, \Phi]_{1/2}(X)=(\phi_{1/2}(X^\bullet), \Phi_{1/2}(X^\times))$. If $e \in \Phi_{1/2}(X^\times)$, we need to prove that $v(e) \subseteq \phi_{1/2}(X^\bullet)$.
\begin{eqnarray*}
e \in \Phi_{1/2}(X^\times) & \Rightarrow & e \in \epsilon^\times \circ \delta^\bullet(X^\times)\\
& \Rightarrow & e \in \lbrace e_{i}, i \in I|v(e_{i}) \subseteq \delta^\bullet(X^\times) \rbrace\\
& \Rightarrow & v(e) \subseteq \delta^\bullet(X^\times) \\
\end{eqnarray*}
But $\delta^\bullet(X^\times)=\underset{e \in X^\times} \cup v(e) \subseteq X^\bullet$. Since $\epsilon^\bullet \circ \delta^\times $ is a closing, $X^\bullet \subseteq \epsilon^\bullet \circ \delta^\times(X^\bullet)$. Therefore $\delta^\bullet(X^\times) \subseteq \epsilon^\bullet \circ \delta^\times(X^\bullet)$. Thus 
\begin{eqnarray*}
e \in \Phi_{1/2}(X^\times) & \Rightarrow & v(e) \subseteq \epsilon^\bullet \circ \delta^\times(X^\bullet)\\
& \Rightarrow &v(e) \subseteq \phi_{1/2}(X^\bullet)
\end{eqnarray*}
Thus $[\phi, \Phi]_{1/2}(X) \in \mathcal{H}$. Hence $\mathcal{H}$ is closed under $[\phi, \Phi]_{1/2}$.

\item[3.] In order to prove that $[\gamma, \Gamma]_{1/2}$ is an opening, it is enough to show that $[\gamma, \Gamma]_{1/2}$ is increasing and idempotent on $\mathcal{H}$. Let $X,Y \in \mathcal{H}$ be such that $X \subseteq Y \subseteq H$. Thus $X^\bullet \subseteq Y^\bullet \subseteq H^\bullet$ and $X^\times \subseteq Y^\times \subseteq H^\times$. We have $[\gamma, \Gamma]_{1/2}(X) = (\gamma_{1/2}(X^\bullet), \Gamma_{1/2}(X^\times))$. Therefore 
\begin{eqnarray*}
[\gamma, \Gamma]_{1/2} \circ [\gamma, \Gamma]_{1/2}(X) &=& [\gamma, \Gamma]_{1/2} \circ (\gamma_{1/2}(X^\bullet), \Gamma_{1/2}(X^\times))\\
&=& (\gamma_{1/2} \circ \gamma_{1/2}(X^\bullet), \Gamma_{1/2} \circ \Gamma_{1/2}(X^\times))\\
&=& (\gamma_{1/2}(X^\bullet), \Gamma_{1/2}(X^\times)) \text{ (since $\gamma_{1/2}$ and $\Gamma_{1/2}$ are openings)} \\
&=& [\gamma, \Gamma]_{1/2}(X)
\end{eqnarray*}
Therefore $[\gamma, \Gamma]_{1/2}$ is idempotent.\\
Since $\gamma_{1/2}$ and $\Gamma_{1/2}$ are openings, $\gamma_{1/2}(X^\bullet) \subseteq \gamma_{1/2}(Y^\bullet)$ and $\Gamma_{1/2}(X^\times) \subseteq \Gamma_{1/2}(Y^\times)$. Thus $[\gamma, \Gamma]_{1/2}(X) \subseteq [\gamma, \Gamma]_{1/2}(Y)$ so that $[\gamma, \Gamma]_{1/2}$ is increasing on $\mathcal{H}$. Hence $[\gamma, \Gamma]_{1/2}$ is an opening. Similarly, we can prove that $[\gamma, \Gamma]_{1}$ is an opening, and that $[\phi, \Phi]_{1/2}$ and $[\phi, \Phi]_{1}$ are closings.
\end{enumerate}
\end{proof}
\subsection{Granulometries}
Granulometries \cite{cousty2013morphological}, \cite{najman2013mathematical} deal with families of openings and closings that are parametrized by a positive number. A family $\lbrace \gamma_{\lambda} \rbrace$ of mappings from a lattice $\mathcal{L} \rightarrow \mathcal{L}$, depending on a positive parameter $ \lambda $ is a granulometry when
\begin{itemize}
\item[(i)] $\gamma_{\lambda}$ is an opening $\forall \lambda \geq 0$
\item[(ii)] $\lambda \geq \mu \geq 0 \Rightarrow  \gamma_{\lambda} \geq \gamma_{\mu}$
\end{itemize}
These conditions are called Matheron's axioms for granulometry.
\begin{defn}
Let $\lambda \in N$. We define $[\gamma, \Gamma]_{\lambda/2}$ (resp. $[\phi, \Phi]_{\lambda/2}$) as follows. $[\gamma, \Gamma]_{\lambda/2} = [\delta, \Delta]^i \circ ([\gamma, \Gamma]_{1/2})^j \circ [\epsilon, \varepsilon]^i$, where $i$ and $j$ are respectively the quotient and remainder when $\lambda$ is divided by 2. 
\end{defn}
\begin{Theorem}
The families $\lbrace [\gamma, \Gamma]_{\lambda/2} | \lambda \in N \rbrace$ and $\lbrace [\phi, \Phi]_{\lambda/2} | \lambda \in N \rbrace$ are granulometries.
\end{Theorem}
\begin{proof}
We know that if $(\alpha, \beta)$ is an adjunction, then $(\alpha \circ \alpha, \beta \circ \beta)$ is an adjunction and so $(\alpha^i, \beta^i)$ is also an adjunction for every $i \in N$ \cite{heijmans1990algebraic}. Now $([\epsilon, \varepsilon]^i, [\delta, \Delta]^i)$ is an adjunction for every $i \in N$, since $([\epsilon, \varepsilon], [\delta, \Delta])$ is an adjunction on $\mathcal{H}$. This implies $[\delta, \Delta]^i \circ [\epsilon, \varepsilon]^i$ is an opening. But $[\delta, \Delta]^i \circ [\epsilon, \varepsilon]^i = [\gamma, \Gamma]_{2i/2}$. This implies $[\gamma, \Gamma]_{\lambda/2}$ is an opening if $\lambda$ is even. If $\lambda$ is odd ($i.e.$ if $\lambda = (2i+1)/2$) then $[\gamma, \Gamma]_{\lambda/2} = [\gamma, \Gamma]_{(2i+1)/2} = [\delta, \Delta]^i \circ [\gamma, \Gamma]_{1/2} \circ [\epsilon, \varepsilon]^i$. If $(\alpha, \beta)$ is an adjunction and $\gamma$ is an opening on a lattice $\mathcal{L}$, then $\beta \circ \gamma \circ \alpha$ is an opening \cite{heijmans1997composing}. Since $([\epsilon, \varepsilon]^i, [\delta, \Delta]^i)$ is an adjunction and $[\gamma, \Gamma]_{1/2}$ is an opening, we have $[\delta, \Delta]^i \circ [\gamma, \Gamma]_{1/2} \circ [\epsilon, \varepsilon]^i$ is an opening.

Now we need to prove that $[\gamma, \Gamma]_{\mu/2}(X) \subseteq [\gamma, \Gamma]_{\lambda/2}(X)$, for $\lambda \leq \mu, \lambda,\mu \in N$ and $X \in \mathcal{H}$. We have $[\gamma, \Gamma]_{1}(X) \subseteq [\gamma, \Gamma]_{1/2}(X) \subseteq X$, for every $X \in \mathcal{H}$. (By property \ref{Property6}). Replacing $X$ 	with $[\epsilon, \varepsilon]^i(X)$, we have $[\gamma, \Gamma]_{1} \circ [\epsilon, \varepsilon]^i(X) \subseteq [\gamma, \Gamma]_{1/2} \circ [\epsilon, \varepsilon]^i(X) \subseteq [\epsilon, \varepsilon]^i(X)$. Now $ [\delta, \Delta]^i$ is a dilation, because $([\epsilon, \varepsilon]^i, [\delta, \Delta]^i)$ is an adjunction. This implies $[\delta, \Delta]^i \circ [\gamma, \Gamma]_{1} \circ [\epsilon, \varepsilon]^i(X) \subseteq [\delta, \Delta]^i \circ [\gamma, \Gamma]_{1/2} \circ [\epsilon, \varepsilon]^i(X) \subseteq [\delta, \Delta]^i \circ [\epsilon, \varepsilon]^i(X)$. That is $[\delta, \Delta]^{i+1} \circ [\epsilon, \varepsilon]^{i+1}(X) \subseteq [\delta, \Delta]^i \circ [\gamma, \Gamma]_{1/2} \circ [\epsilon, \varepsilon]^i(X) \subseteq [\delta, \Delta]^i \circ [\epsilon, \varepsilon]^i(X)$. Hence $[\gamma, \Gamma]_{(2i+2)/2}(X) \subseteq [\gamma, \Gamma]_{(2i+1)/2}(X) \subseteq [\gamma, \Gamma]_{2i/2}(X)$ for every $i \in N$. This implies $[\gamma, \Gamma]_{\mu/2}(X) \subseteq [\gamma, \Gamma]_{\lambda/2}(X)$ for every $\mu \geq \lambda$, $\mu, \lambda \in N$. Therefore the family $\lbrace [\gamma, \Gamma]_{\lambda/2} | \lambda \in N \rbrace$ is a granulometry. Similar line of arguments can be used to prove that $\lbrace [\phi, \Phi]_{\lambda/2} | \lambda \in N \rbrace$ is a granulometry.
\end{proof}

\begin{defn}
Let $\lambda \in N$ and $X \in \mathcal{H}$. We define the operator $ASF_{\lambda/2}$ by 
\begin{equation*}
ASF_{\lambda/2}(X)= 
\begin{cases}
 X & \text{if $\lambda = 0$}\\
[\gamma, \Gamma]_{\lambda/2} \circ [\phi, \Phi]_{\lambda/2} \circ ASF_{(\lambda-1)/2}(X) & \text{if $\lambda \neq 0$}
\end{cases}
\end{equation*}
\end{defn}
\begin{Corollary}
The family  $\lbrace ASF_{\lambda/2} | \lambda \in N \rbrace$ is a family of alternate sequential filters.
\end{Corollary}

%\section{Extension to weighted hypergraphs}

\section{Conclusion}
This paper investigates the lattice of all subhypergraphs of a hypergraph $H$ and provides it with morphological operators. We created new dilations, erosions, openings, closings, granulometries and alternate sequential filters whose input and output are both hypergraphs. The proposed framework is then extended from subgraphs to functions that weight the vertices and hyperedges of a hypergraph. The proposed framework of hypergraphs presented in this paper encompasses mathematical morphology on graphs, simplicial complexes and morphology based on discrete spatially variant structuring elements.

%This paper investigates the lattice of all subhypergraphs of a hypergraph $H$ and provides it with morphological operators. By the composition of the operators presented in this paper, we can define other adjunctions on hypergraphs. The proposed framework can be extended to morphological filtering on hypergraphs and is left to the future work.

\bibliographystyle{plain}
\bibliography{bino-bib}

\begin{thebibliography}{10}

\bibitem{berge1989hypergraphs}
Claude Berge.
\newblock Hypergraphs: Combinatorics of finite sets, 1989.

\bibitem{bloch2011mathematical}
Isabelle Bloch and Alain Bretto.
\newblock Mathematical morphology on hypergraphs: Preliminary definitions and
  results.
\newblock In {\em Discrete Geometry for Computer Imagery}, pages 429--440.
  Springer, 2011.

\bibitem{bloch2013mathematical}
Isabelle BLOCH and Alain BRETTO.
\newblock Mathematical morphology on hypergraphs, application to similarity and
  positive kernel.
\newblock {\em Computer vision and image understanding}, 117(4):342--354, 2013.

\bibitem{bloch2013similarity}
Isabelle Bloch, Alain Bretto, and Aur{\'e}lie Leborgne.
\newblock Similarity between hypergraphs based on mathematical morphology.
\newblock In {\em Mathematical Morphology and Its Applications to Signal and
  Image Processing}, pages 1--12. Springer, 2013.

\bibitem{bretto1997combinatorics}
Alain Bretto, J~Azema, Hocine Cherifi, and Bernard Laget.
\newblock Combinatorics and image processing.
\newblock {\em Graphical Models and Image Processing}, 59(5):265--277, 1997.

\bibitem{cousty2013morphological}
Jean Cousty, Laurent Najman, Fabio Dias, and Jean Serra.
\newblock Morphological filtering on graphs.
\newblock {\em Computer Vision and Image Understanding}, 117(4):370--385, 2013.

\bibitem{cousty2009some}
Jean Cousty, Laurent Najman, and Jean Serra.
\newblock Some morphological operators in graph spaces.
\newblock In {\em Mathematical Morphology and Its Application to Signal and
  Image Processing}, pages 149--160. Springer, 2009.

\bibitem{dharmarajan2010hypergraph}
R~Dharmarajan and K~Kannan.
\newblock A hypergraph-based algorithm for image restoration from salt and
  pepper noise.
\newblock {\em AEU-International Journal of Electronics and Communications},
  64(12):1114--1122, 2010.

\bibitem{dias2011some}
F{\'a}bio Dias, Jean Cousty, and Laurent Najman.
\newblock Some morphological operators on simplicial complex spaces.
\newblock In {\em Discrete Geometry for Computer Imagery}, pages 441--452.
  Springer, 2011.

\bibitem{gonzalez2002woods}
Rafael~C Gonzalez and E~Richard.
\newblock Woods, digital image processing.
\newblock {\em ed: Prentice Hall Press, ISBN 0-201-18075-8}, 2002.

\bibitem{heijmans1993graph}
Henk Heijmans and L~Vincent.
\newblock Graph morphology in image analysis.
\newblock 1993.

\bibitem{heijmans1997composing}
Henk~JAM Heijmans.
\newblock Composing morphological filters.
\newblock {\em Image Processing, IEEE Transactions on}, 6(5):713--723, 1997.

\bibitem{heijmans1990algebraic}
Henk~JAM Heijmans and Christian Ronse.
\newblock The algebraic basis of mathematical morphology i. dilations and
  erosions.
\newblock {\em Computer Vision, Graphics, and Image Processing},
  50(3):245--295, 1990.

\bibitem{najman2012short}
Laurent Najman, Fernand Meyer, et~al.
\newblock A short tour of mathematical morphology on edge and vertex weighted
  graphs.
\newblock {\em Image Processing and Analysis with Graphs: Theory and Practice},
  pages 141--174, 2012.

\bibitem{najman2013mathematical}
Laurent Najman and Hugues Talbot.
\newblock {\em Mathematical Morphology}.
\newblock John Wiley \& Sons, 2013.

\bibitem{ronse1990mathematical}
Christian Ronse.
\newblock Why mathematical morphology needs complete lattices.
\newblock {\em Signal processing}, 21(2):129--154, 1990.

\bibitem{serra1988image}
Jean Serra.
\newblock Image analysis and mathematical morphology. vol. 2: Theoretical
  advances (edited by j. serra), chapter 5, introduction to morphological
  filters, 1988.

\bibitem{serra1982image}
Jean~Paul Serra.
\newblock Image analysis and mathematical morphology, 1982.

\bibitem{shih2010image}
Frank~Y Shih.
\newblock {\em Image processing and mathematical morphology: fundamentals and
  applications}.
\newblock CRC press, 2010.

\bibitem{stell2012relations}
John~G Stell.
\newblock Relations on hypergraphs.
\newblock In {\em Relational and Algebraic Methods in Computer Science}, pages
  326--341. Springer, 2012.

\bibitem{vincent1989graphs}
Luc Vincent.
\newblock Graphs and mathematical morphology.
\newblock {\em Signal Processing}, 16(4):365--388, 1989.

\end{thebibliography}
\end{document}